\documentclass[11pt,finalcls,onecolumn]{IEEEtran}
\usepackage{graphicx}
\usepackage{epstopdf}
\usepackage{epsfig}
\usepackage{amsmath}
\usepackage{amssymb}
\usepackage{mathrsfs}
\usepackage{mathtools}
\usepackage{float}
\usepackage{url}
\usepackage{enumerate}

\usepackage[left=20mm,right=20mm,top=24mm,bottom=24mm]{geometry}

\newtheorem{corollary}{Corollary}
\newtheorem{definition}{Definition}

\newtheorem{theorem}{Theorem}
\newtheorem{proposition}{Proposition}
\newtheorem{lemma}{Lemma}
\newtheorem{remark}{Remark}
\newtheorem{example}{Example}

\DeclareMathOperator*{\esssup}{\mathrm{ess } \sup}

\newcommand{\expectation}{\mathbb{E}}
\newcommand{\prob}{\mathbb{P}}
\newcommand{\naturals}{\{1, 2, 3, \ldots \}}
\newcommand{\Reals}{\mathbb{R}}
\newcommand{\set}{\mathcal}
\newcommand{\supp}{\mathop{\mathrm{supp}}}
\newcommand{\overbar}[1]{\mkern 1.5mu\overline{\mkern-1.5mu#1\mkern-1.5mu}\mkern 1.5mu}

\def \PZ  { P_{{\mathtt{0}}} }
\def \PU  { P_{{\mathtt{1}}} }
\def \HZ  {\mathsf{H}_{{\mathtt{0}}} }
\def \HU  {\mathsf{H}_{{\mathtt{1}}} }
\def \RZ  { R_{{\mathtt{0}}} }

\begin{document}

\title{Arimoto-R\'{e}nyi Conditional Entropy\\ and Bayesian $M$-ary
Hypothesis Testing}
\author{Igal Sason and Sergio Verd\'{u}
\thanks{

\vspace*{5.2cm}
I. Sason is with the Andrew and Erna Viterbi Faculty of Electrical Engineering,
Technion--Israel Institute of Technology, Haifa 32000, Israel (e-mail:
sason@ee.technion.ac.il).}
\thanks{
S. Verd\'{u} is with the Department of Electrical Engineering, Princeton
University, Princeton, New Jersey 08544, USA (e-mail: verdu@princeton.edu).}
\thanks{
This manuscript has been submitted to the {\em IEEE Transactions on Information Theory}
in September~27, 2016, revised in May~24, 2017, and accepted for publication in September~4, 2017.
It has been presented in part at the {\em 2017 IEEE International Symposium
on Information Theory}, Aachen, Germany, June 25--30, 2017.}
\thanks{
This work has been supported by the Israeli Science Foundation (ISF) under
Grant 12/12, by ARO-MURI contract number  W911NF-15-1-0479
and in part by the Center for Science of Information, an NSF Science and
Technology Center under Grant CCF-0939370.}
}

\maketitle
\begin{abstract}
\normalsize
This paper gives upper and lower bounds on the minimum error probability of
Bayesian $M$-ary hypothesis testing in terms of the Arimoto-R\'enyi conditional
entropy of an arbitrary order $\alpha$. The improved tightness of these bounds
over their specialized versions with the Shannon conditional entropy ($\alpha=1$)
is demonstrated. In particular, in the case where $M$ is finite, we show how to
generalize Fano's inequality under both the conventional and list-decision
settings. As a counterpart to the generalized Fano's inequality, allowing $M$
to be infinite, a lower bound on the Arimoto-R\'enyi conditional entropy is derived
as a function of the minimum error probability. Explicit upper and lower
bounds on the minimum error probability are obtained as a function of the
Arimoto-R\'enyi conditional entropy for both positive and negative $\alpha$.
Furthermore, we give upper bounds on the minimum error probability
as functions of the R\'enyi divergence.
In the setup of discrete memoryless channels, we analyze the exponentially
vanishing decay of the Arimoto-R\'enyi conditional entropy of the transmitted
codeword given the channel output when averaged over a random-coding ensemble.
\end{abstract}

\eject

\section{Introduction}\label{section:1}
In Bayesian $M$-ary hypothesis testing, we have:
\begin{itemize}
\item
$M$ possible explanations, hypotheses or models for the
$\set{Y}$-valued data $\{ P_{Y|X=m}, m \in \set{X} \}$
where the set of model indices satisfies $|\set{X}| = M$; and
\item
a prior distribution $P_X$ on $\set{X}$.
\end{itemize}
The minimum probability of error of $X$ given $Y$, denoted by
$\varepsilon_{X|Y}$, is achieved by the \textit{maximum-a-posteriori}
(MAP) decision rule. Summarized in Sections~\ref{subsec: 1.1}
and~\ref{subsec: 1.3}, a number of bounds on $\varepsilon_{X|Y}$ involving
Shannon information measures or their generalized R\'{e}nyi measures
have been obtained in the literature. As those works attest, there is
considerable motivation for the study of the relationships between
error probability and information measures. The minimum error probability
of Bayesian $M$-ary hypothesis testing is rarely directly computable,
and the best lower and upper bounds are information theoretic. Furthermore,
their interplay is crucial in the proof of coding theorems.

\subsection{Existing bounds involving Shannon information measures}
\label{subsec: 1.1}
\begin{enumerate}
\item\label{i:1}
Fano's inequality \cite{Fano52} gives an upper bound on the conditional
entropy $H(X|Y)$ as a function of $\varepsilon_{X|Y}$ when $M$ is finite.
\item\label{i:2}
Shannon's inequality \cite{Shannon58} (see also \cite{verduITA2011})
gives an explicit lower bound on $\varepsilon_{X|Y}$ as a function of
$H(X|Y)$, also when $M$ is finite.
\item\label{i:3}
Tightening  another bound by Shannon \cite{CES57}, Poor and Verd\'u
\cite{poor1995lower} gave a lower bound on $\varepsilon_{X|Y}$
as a function of the distribution of the conditional information
(whose expected value is $H(X|Y)$). This bound was generalized
by Chen and Alajaji \cite{ChenA12}.
\item\label{i:4}
Baladov\'a \cite{Baladova66}, Chu and Chueh \cite[(12)]{ChuC66},
and Hellman and Raviv \cite[(41)]{HellmanR70} showed that
\begin{align}\label{onehalf}
\varepsilon_{X|Y} \leq \tfrac12\, H(X|Y) ~ \mbox{bits}
\end{align}
when $M$ is finite.
It is also easy to show that (see, e.g., \cite[(21)]{FederM94})
\begin{align}\label{mystery}
\varepsilon_{X|Y} \leq 1 - \exp\bigl(-H(X | Y) \bigr).
\end{align}
Tighter and generalized upper bounds on
$\varepsilon_{X|Y}$ were obtained by Kovalevsky \cite{Kovalevsky},
Tebbe and Dwyer \cite{TebbeD68}, and Ho and Verd\'u \cite[(109)]{HoS-IT10}.
\item\label{i:5}
Based on the fundamental tradeoff of a certain auxiliary binary hypothesis test, Polyanskiy
\textit{et al.} \cite{PPV} gave the \textit{meta-converse} implicit lower bound
on $\varepsilon_{X|Y}$, which for some choice of auxiliary quantities is shown
to be tight in \cite{gonzalo2016}.
\item\label{i:8}
Building up on \cite{HellmanR70}, Kanaya and Han \cite{kanaya1995asymptotics}
showed that in the case of independent identically distributed (i.i.d.) observations,
$\varepsilon_{X|Y^n}$ and $H(X|Y^n)$ vanish exponentially at the same speed, which is
governed by the Chernoff information between the closest hypothesis pair.
In turn,  Leang and Johnson \cite{LeangJ97} showed that the same exponential decay
holds for those cost functions that have zero cost for correct decisions.
\item\label{i:9}
Birg\'{e} \cite{birge2005} gave an implicit lower bound on the minimax error
probability ($\varepsilon_{X|Y}$ maximized over all possible priors) as a function
of the pairwise relative entropies among the various models.
\item\label{i:10-a}
Generalizing Fano's inequality, Han and Verd\'u \cite{TSV94} gave lower bounds
on the mutual information $I(X;Y)$ as a function of $\varepsilon_{X|Y}$.
\item\label{i:11}
Grigoryan {\em et al.} \cite{GHVK11} and Sason \cite{Sason12} obtained error bounds,
in terms of relative entropies, for hypothesis testing with a rejection option. Such
bounds recently proved useful in the context
of an almost-fixed size hypothesis decision algorithm, which bridges the gap in performance
between fixed-sample and sequential hypothesis testing \cite{LalitaJ16}.
\end{enumerate}

\subsection{R\'{e}nyi's information measures}
\label{subsec: 1.2}
In this paper, we give upper and lower bounds on $\varepsilon_{X|Y}$ not in terms of
$H(X|Y)$ but in terms of the \textit{Arimoto-R\'enyi} conditional entropy $H_\alpha (X|Y)$
of an arbitrary order $\alpha$. Loosening the axioms given by Shannon \cite{CES48} as
a further buttress for $H(X)$, led R\'enyi \cite{Renyientropy} to the introduction of the
\textit{R\'enyi entropy} of order $\alpha \in [0, \infty]$, $H_\alpha (X)$, as well as
the \textit{R\'enyi divergence} $D_{\alpha} (P \| Q )$. R\'enyi's entropy and divergence
coincide with Shannon's and Kullback-Leibler's measures, respectively, for $\alpha=1$.
Among other applications, $H_\alpha (X)$ serves to analyze the fundamental limits of
lossless data compression \cite{campbell65,courtadeverdu2014-1,Csiszar95}. Unlike
Shannon's entropy, $H_\alpha (X)$ with $\alpha \neq 1$ suffers from the disadvantage
that the inequality $H_\alpha (X_1, X_2) \leq H_\alpha (X_1) + H_\alpha (X_2)$
does not hold in general. Moreover, if we mimic the definition of $H(X|Y)$ and define
conditional R\'enyi entropy as $\sum_{y \in \mathcal{Y}} P_Y(y) \, H_{\alpha}(X | Y=y)$,
we find the unpleasant property that the conditional version may be larger than
$H_\alpha (X)$. To remedy this situation, Arimoto \cite{Arimoto75} introduced a
notion of conditional R\'enyi entropy, which in this paper we denote as $H_\alpha (X|Y)$,
and which is indeed upper bounded by $H_\alpha (X)$. The corresponding
$\alpha$-\textit{capacity} $\underset{X}{\max} \{ H_\alpha (X) - H_\alpha (X|Y) \}$ yields
a particularly convenient expression for the converse bound found by Arimoto
\cite{Arimoto73} (see also \cite{PolyanskiyV10}). The Arimoto-R\'enyi conditional entropy
has also found applications in guessing and secrecy problems with side information
(\cite{Arikan96}, \cite{BracherHL15}, \cite{HayashiT17}, \cite{HayashiT16} and
\cite{Sundaresan07}), sequential decoding \cite{Arikan96}, task encoding with side
information available to both the task-describer (encoder) and the task-performer
\cite{BunteL14a}, and the list-size capacity of discrete memoryless channels in
the presence of noiseless feedback~\cite{BunteL14b}. In \cite{SakaiIwata16},
the joint range of $H(X|Y)$ and $H_{\alpha}(X|Y)$ is obtained when the random variable
$X$ is restricted to take its values on a finite set with a given cardinality.
Although outside the scope of this paper, Csisz\'ar \cite{Csiszar95} and Sibson
\cite{Sibson69} proposed implicitly other definitions of conditional R\'enyi entropy,
which lead to the same value of the $\alpha$-capacity and have found various applications
(\cite{Csiszar95} and \cite{Verdu-ITA15}). The quantum generalizations of R\'{e}nyi
entropies and their conditional versions
(\cite{HayashiT17, HayashiT16, HayashiTo15, Iwamoto14, KonigWY, Muller-Lenert13, Renner05, TomamichelBH14})
have recently served to prove strong converse theorems in
quantum information theory (\cite{HayashiT16, Leditzky16} and \cite{WildeWY-14}).

\subsection{Existing bounds involving R\'{e}nyi's information measures and related measures}
\label{subsec: 1.3}
In continuation to the list of existing bounds involving the
conditional entropy and relative entropy in Section~\ref{subsec: 1.1},
bounds involving R\'{e}nyi's information measures and related measures include:
\begin{enumerate}
\setcounter{enumi}{9}
\item\label{i:6}
For Bayesian binary hypothesis testing, Hellman and Raviv \cite{HellmanR70} gave an upper
bound on $\varepsilon_{X|Y}$ as a function of the prior probabilities and the
R\'enyi divergence of order~$\alpha \in [0,1]$ between the two models. The special
case of $\alpha = \frac12$ yields the \textit{Bhattacharyya bound} \cite{Kailath67}.
\item\label{i:7}
In \cite{Devijver74} and \cite{Vajda68}, Devijver and Vajda derived  upper and
lower bounds on $\varepsilon_{X|Y}$ as a function of the quadratic Arimoto-R\'enyi
conditional entropy $H_2(X|Y)$.
\item\label{i:Cover-Hart}
In \cite{Cover_Hart67}, Cover and Hart showed that if $M$ is finite, then
\begin{align} \label{eq: CoverH67}
\varepsilon_{X|Y} \leq H_{\mathrm{o}}(X|Y) \leq \varepsilon_{X|Y} \left( 2 - \tfrac{M}{M-1} \; \varepsilon_{X|Y} \right)
\end{align}
where
\begin{align} \label{eq: quad. ent}
H_{\mathrm{o}}(X|Y) = \mathbb{E} \left[ \, \sum_{x \in \set{X}} P_{X|Y}(x|Y) \left(1-P_{X|Y}(x|Y) \right) \right]
\end{align}
is referred to as conditional quadratic entropy. The bounds in \eqref{eq: CoverH67}
coincide if $X$ is equiprobable on $\set{X}$, and $X$ and $Y$ are independent.
\item\label{i:10}
One of the bounds by Han and Verd\'u \cite{TSV94} in Item~\ref{i:10-a}) was
generalized by Polyanskiy and Verd\'u \cite{PolyanskiyV10} to give a
lower bound on the $\alpha$-mutual information (\cite{Sibson69,Verdu-ITA15}).
\item\label{i:12}
In \cite{Shayevitz_ISIT11}, Shayevitz gave a lower bound, in terms of
the R\'enyi divergence, on the maximal worst-case miss-detection exponent
for a binary composite hypothesis testing problem when the false-alarm probability
decays to zero with the number of i.i.d. observations.
\item\label{i:13}
Tomamichel and Hayashi studied optimal exponents of binary composite hypothesis
testing, expressed in terms of R\'enyi's information measures (\cite{HayashiTo15, TomamichelH16}).
A measure of dependence was studied in \cite{TomamichelH16} (see also
Lapidoth and Pfister \cite{LapidothP16}) along with its role in composite
hypothesis testing.
\item\label{i:14}
Fano's inequality was generalized by Toussiant in \cite[Theorem~2]{Toussiant77}
(see also \cite{ErdogmusP04} and \cite[Theorem~6]{Iwamoto14}) not with the
Arimoto-R\'enyi conditional entropy but with the average of the R\'{e}nyi
entropies of the conditional distributions (which, as we mentioned in
Section~\ref{subsec: 1.2}, may exceed the unconditional R\'{e}nyi entropy).
\item\label{i:9-b}
A generalization of the minimax lower bound by Birg\'{e} \cite{birge2005} in
Item~\ref{i:9}) to lower bounds involving $f$-divergences has been studied
by Guntuboyina \cite{Guntuboyina11}.
\end{enumerate}

\subsection{Main results and paper organization}
\label{subsec: 1.4}

It is natural to consider generalizing the relationships between
$\varepsilon_{X|Y}$ and $H(X|Y)$, itemized in Section~\ref{subsec: 1.1},
to the Arimoto-R\'enyi conditional entropy
$H_\alpha (X|Y)$. In this paper we find pleasing counterparts to
the bounds in Items \ref{i:1}), \ref{i:4}), \ref{i:8}), \ref{i:6}),
\ref{i:7}), and \ref{i:10}), resulting in generally
tighter bounds. In addition, we enlarge the scope of the problem to consider
not only $\varepsilon_{X|Y}$ but the probability that a list decision rule
(which is allowed to output a set of $L$ hypotheses) does not include the
true one. Previous work on extending Fano's inequality to the setup of list
decision rules includes \cite[Section~5]{Ahlswede_Gacs_Korner},
\cite[Lemma~1]{Kim_Sutivong_Cover}, \cite[Appendix~3.E]{FnT14} and
\cite[Chapter~5]{verdubook}.

Section~\ref{section: notation and definitions} introduces the basic
notation and definitions of R\'enyi information measures.
Section~\ref{section: tight bound and joint range of RE} finds bounds
for the guessing problem in which there are no observations; those bounds,
which prove to be instrumental in the sequel, are tight in the sense that
they are attained with equality for certain random variables.
Section~\ref{section: cond. RE vs. Pe} contains
the main results in the paper on the interplay between $\varepsilon_{X|Y}$
and $H_\alpha (X|Y)$, giving counterparts to a number of those existing
results mentioned in Sections~\ref{subsec: 1.1} and~\ref{subsec: 1.3}.
In particular:
\begin{enumerate}
\setcounter{enumi}{17}
\item \label{s:1}
An upper bound on $H_\alpha (X|Y)$ as a function of $\varepsilon_{X|Y}$ is
derived for positive $\alpha$ (Theorem~\ref{theorem: generalized Fano inequality});
it provides an implicit lower bound on $\varepsilon_{X|Y}$ as a function
of $H_\alpha (X|Y)$;
\item \label{s:2}
Explicit lower bounds on $\varepsilon_{X|Y}$ are given as a function of
$H_\alpha (X|Y)$ for both positive and negative $\alpha$
(Theorems~\ref{proposition: closed-form lower bound on epsilon},
\ref{theorem: LB via HolderI}, and~\ref{theorem: LB via RevHolderI});
\item \label{s:3}
The lower bounds in Items~\ref{s:1}) and~\ref{s:2}) are generalized to
the list-decoding setting
(Theorems~\ref{theorem: generalized Fano-Renyi - list decoding},
\ref{theorem: LB via HolderI - list decoding},
and~\ref{theorem: LB via RevHolderI - list decoding});
\item \label{s:4}
As a counterpart to the generalized Fano's inequality, we derive a lower
bound on $H_\alpha (X|Y)$ as a function of $\varepsilon_{X|Y}$ capitalizing
on the Schur concavity of R\'enyi entropy
(Theorem~\ref{theorem: LB on the conditional Renyi entropy});
\item \label{s:5}
Explicit upper bounds on $\varepsilon_{X|Y}$ as a function of
$H_\alpha (X|Y)$ are obtained (Theorem~\ref{theorem: UB min err});
\end{enumerate}
\par
Section~\ref{section: BHT} gives explicit upper bounds on $\varepsilon_{X|Y}$
as a function of the error probability of associated binary tests
(Theorem~\ref{thm:UBMbasedon2HT}) and of the R\'{e}nyi divergence
(Theorems~\ref{theorem: generalized Hellman-Raviv bound}
and~\ref{thm:UBSV20160616}) and Chernoff information
(Theorem~\ref{thm:UBSV20160616}).
\par
Section~\ref{section: random coding} analyzes the exponentially vanishing decay
of the Arimoto-R\'enyi conditional entropy of the transmitted codeword given the
channel output when averaged over a code ensemble (Theorems~\ref{theorem: random coding}
and~\ref{theorem: Ralpha - MBIOS}). Concluding remarks are given in
Section~\ref{sec: conclusions}.

\section{R{\'e}nyi Information Measures: Definitions and Basic Properties}
\label{section: notation and definitions}

\begin{definition} \cite{Renyientropy} \label{definition: Renyi entropy}
Let $P_X$  be a probability distribution on a discrete set $\set{X}$.
The \textit{R\'{e}nyi entropy of order} $\alpha \in (0,1) \cup (1, \infty)$
of $X$, denoted by $H_{\alpha}(X)$ or $H_{\alpha}(P_X)$, is defined as
\begin{align} \label{eq: Renyi entropy}
H_{\alpha}(X) & = \frac1{1-\alpha} \, \log
\sum_{x \in \set{X}} P_X^{\alpha}(x) \\
\label{eq2: Renyi entropy}
& = \frac{\alpha}{1-\alpha} \, \log \| P_X \|_{\alpha}
\end{align}
where $\|P_X\|_{\alpha} = \left( \sum_{x \in \set{X}}
P_X^{\alpha}(x) \right)^{\frac1\alpha}$.
By its continuous extension,
\begin{align}
\label{eq: RE of zero order}
& H_0(X) = \log \, \bigl| \supp P_X \bigr|, \\[0.1cm]
\label{eq: Shannon entropy}
& H_1(X) = H(X), \\[0.1cm]
\label{eq: RE of infinite order}
& H_{\infty}(X) = \log\frac1{p_{\max}}
\end{align}
where $\supp P_X = \{x \in \set{X} \colon P_X(x) > 0 \}$
is the support of $P_X$,
and $p_{\max}$ is the largest of the masses of $X$.
\end{definition}

When guessing the value of $X$ in the absence of an observation
of any related quantity, the maximum probability of success is
equal to $p_{\max}$; it is achieved by guessing that $X = x_0$
where $x_0$ is a {\em mode} of $P_X$, i.e., $P_X(x_0)=p_{\max}$.
Let
\begin{align} \label{eq: blind decision}
& \varepsilon_X = 1-p_{\max}
\end{align}
denote the minimum probability of error in guessing the value
of $X$. From \eqref{eq: RE of infinite order} and \eqref{eq: blind decision},
the minimal error probability in guessing $X$ without side information
is given by
\begin{align} \label{eq2: min error prob.}
\varepsilon_X = 1 - \exp\bigl(-H_{\infty}(X)\bigr).
\end{align}

\begin{definition} \label{definition: binary RE}
For $\alpha \in (0,1) \cup (1, \infty)$, the {\em binary
R\'{e}nyi entropy of order $\alpha$} is the function
$h_{\alpha} \colon [0,1] \to [0, \log 2]$ that is
defined, for $p \in [0,1]$, as
\begin{equation}
\label{eq1: binary RE}
h_{\alpha}(p) = H_\alpha(X_p) = \frac1{1-\alpha} \,
\log\bigl(p^\alpha + (1-p)^{\alpha} \bigr),
\end{equation}
where $X_p$ takes two possible values with probabilities $p$ and $1-p$.
The continuous extension of the binary
R\'{e}nyi entropy at $\alpha=1$ yields the binary entropy function:
\begin{align}
h(p) &=  p \log \frac1p + (1-p) \log \frac1{1-p}
\label{eq3: binary RE}
\end{align}
for $p \in (0,1)$, and $h(0)=h(1)=0$.
\end{definition}

In order to put forth generalizations of Fano's
inequality and bounds on the error probability,
we consider Arimoto's proposal for the conditional
R\'{e}nyi entropy (named, for short, the
Arimoto-R\'{e}nyi conditional entropy).
\begin{definition} \cite{Arimoto75}
\label{definition: AR conditional entropy}
Let $P_{XY}$ be defined on $\set{X} \times \set{Y}$,
where $X$ is a discrete random variable.
The \textit{Arimoto-R\'{e}nyi conditional entropy of order
$\alpha \in [ 0, \infty]$} of $X$ given $Y$ is defined as follows:
\begin{itemize}
\item
If $\alpha \in (0,1) \cup (1, \infty) $, then
\begin{align}
\label{eq1: Arimoto - cond. RE}
H_{\alpha}(X | Y) &= \frac{\alpha}{1-\alpha} \,
\log \, \mathbb{E} \left[
\left( \, \sum_{x \in \set{X}} P_{X|Y}^{\alpha}(x|Y)
\right)^{\frac1{\alpha}} \right] \\
&= \frac{\alpha}{1-\alpha} \, \log \mathbb{E}
\left[ \| P_{X|Y} ( \cdot | Y ) \|_\alpha \right] \\
\label{eq2: Arimoto - cond. RE}
&= \frac{\alpha}{1-\alpha} \, \log
\, \sum_{y \in \set{Y}} P_Y(y) \, \exp \left(
\frac{1-\alpha}{\alpha} \;
H_{\alpha}(X | Y=y) \right),
\end{align}
where \eqref{eq2: Arimoto - cond. RE}
applies if $Y$ is a discrete random variable.
\item By its continuous extension, the Arimoto-R\'enyi conditional
entropy of orders $0$, $1$, and $\infty$ are defined as
\begin{align}
\label{eq: cond. RE at 0:gral}
H_0(X|Y) &= \esssup H_0 \left(P_{X|Y} (\cdot | Y) \right) \\[0.1cm]
\label{eq: cond. RE at 0}
&= \log \, \max_{y \in \set{Y}} \,
\bigl| \supp P_{X|Y}(\cdot|y) \bigr| \\[0.1cm]
\label{eq2: cond. RE at 0}
&= \max_{y \in \set{Y}} \, H_0(X \, | \, Y=y), \\[0.1cm]
\label{eq: cond. RE at 1}
H_1(X|Y) &= H(X|Y), \\[0.1cm]
\label{eq: cond. RE at infinity}
H_{\infty}(X|Y) &= \log \, \frac1{\expectation
\Bigl[ \underset{x \in \set{X}}{\max} \, P_{X|Y}(x|Y) \Bigr]}
\end{align}
where $\esssup$ in \eqref{eq: cond. RE at 0:gral} denotes the
essential supremum, and \eqref{eq: cond. RE at 0} and
\eqref{eq2: cond. RE at 0} apply if $Y$ is a discrete random variable.
\end{itemize}
\end{definition}

Although not nearly as important, sometimes in the context
of finitely valued random variables, it is useful to consider
the unconditional and conditional R\'enyi entropies of negative
orders $\alpha \in (-\infty, 0)$ in \eqref{eq: Renyi entropy}
and \eqref{eq1: Arimoto - cond. RE} respectively.
By continuous extension
\begin{align}
\label{eq: RE of order -infinity}
& H_{-\infty}(X) = \log \frac1{p_{\min}}, \\
\label{eq: cond. RE of order -infinity}
& H_{-\infty}(X|Y) = \log \hspace*{0.05cm} \mathbb{E}^{-1} \hspace*{-0.1cm}
\left[ \underset{x \in \set{X}: \, P_{X|Y}(x|Y) > 0}{\min} \, P_{X|Y}(x|Y) \right]
\end{align}
where $p_{\min}$ in \eqref{eq: RE of order -infinity} is the smallest
of the nonzero masses of $X$, and $\mathbb{E}^{-1}[Z]$ in \eqref{eq: cond. RE of order -infinity}
denotes the reciprocal of the expected value of $Z$.
In particular, note that if $|\set{X}| = 2$, then
\begin{align}\label{inf-inf}
H_{-\infty}(X|Y) = \log\frac{1}{1 - \exp \left(-  H_{\infty}(X|Y) \right)}.
\end{align}
Basic properties of $H_\alpha (X|Y)$ appear in \cite{FehrB14}, and
in \cite{Iwamoto14, TomamichelBH14} in the context of quantum information theory. Next,
we provide two additional properties, which are used in the sequel.
\begin{proposition} \label{prop1: mon}
$H_\alpha (X|Y)$ is monotonically decreasing in $\alpha$ throughout the real line.
\end{proposition}

The proof of Proposition~\ref{prop1: mon} is given in Appendix~\ref{app: mon}.
In the special case of $\alpha \in[1, \infty]$ and discrete $Y$, a proof can be
found in \cite[Proposition~4.6]{Berens2013}.

We will also have opportunity to use the following monotonicity result, which
holds as a direct consequence of the monotonicity of the norm of positive orders.
\begin{proposition} \label{prop2: mon}
$\frac{\alpha-1}{\alpha} \, H_\alpha (X|Y)$ is monotonically increasing in $\alpha$
on $(0, \infty)$ and $(-\infty, 0)$.
\end{proposition}
\begin{remark}
Unless $X$ is a deterministic function of $Y$, it follows from \eqref{eq: cond. RE at 0:gral}
that $H_0(X|Y) > 0$; consequently,
\begin{align}
\lim_{\alpha \uparrow 0} \frac{\alpha-1}{\alpha} \, H_\alpha (X|Y) &= +\infty,\\
\lim_{\alpha \downarrow 0} \frac{\alpha-1}{\alpha} \, H_\alpha (X|Y) &= -\infty.
\end{align}
\end{remark}

It follows from Proposition~\ref{prop1: mon} that if $\beta >0$, then
\begin{align}\label{P0N}
H_\beta (X|Y) \leq H_0 (X|Y) \leq H_{-\beta} (X|Y),
\end{align}
and, from Proposition~\ref{prop2: mon},
\begin{align}\label{mon}
\frac{\alpha-1}{\alpha} \, H_\alpha (X|Y) \leq \frac{\beta-1}{\beta} \, H_\beta (X|Y)
\end{align}
if $ 0 < \alpha <\beta$ or $\alpha < \beta <0$.

The third R\'enyi information measure used in this paper is the R\'enyi
divergence. Properties of the R\'enyi divergence are studied in
\cite{ErvenH14}, \cite[Section~8]{SV16} and \cite{Shayevitz_ISIT11}.
\begin{definition} \cite{Renyientropy}
Let $P$ and $Q$ be probability measures on $\set{X}$ dominated by $R$,
and let their densities be respectively denoted by $p = \frac{\mathrm{d}P}{\mathrm{d}R}$
and $q = \frac{\mathrm{d}Q}{\mathrm{d}R}$.
The {\em R\'{e}nyi divergence of order $\alpha \in [ 0, \infty]$} is defined as follows:
\begin{itemize}
\item
If $\alpha \in (0,1) \cup (1, \infty) $, then
\begin{align} \label{eq:RD0}
D_{\alpha}(P\|Q) &=
\frac1{\alpha-1} \; \log \mathbb{E} \left[
p^\alpha(Z) \, q^{1-\alpha}(Z) \right] \\[0.1cm]
&= \frac1{\alpha-1} \; \log \,
\sum_{x \in \set{X}} P^\alpha (x) \,
Q^{1-\alpha} (x) \label{eq:RD1}
\end{align}
where $Z \sim R$ in \eqref{eq:RD0}, and \eqref{eq:RD1} holds if
$\set{X}$ is a discrete set.
\item By the continuous extension of $D_{\alpha}(P \| Q)$,
\begin{align}
\label{eq: d0}
D_0(P \| Q) &= \underset{\set{A}: P(\set{A})=1}{\max} \log \frac1{Q(\set{A})}, \\[0.1cm]
\label{def:d1}
D_1(P\|Q) &= D(P\|Q), \\[0.1cm]
\label{def:dinf}
D_{\infty}(P\|Q) &= \log \, \esssup \frac{p(Z)}{q(Z)}
\end{align}
with $Z\sim R$.
\end{itemize}
\end{definition}

\begin{definition} \label{definition: binary RD}
For all $\alpha \in (0,1) \cup (1, \infty)$, the {\em binary R\'{e}nyi divergence of
order $\alpha$}, denoted by $d_{\alpha}(p\|q)$ for $(p,q) \in [0,1]^2$,
is defined as $D_{\alpha}([p ~1-p] \| [q ~1-q])$. It is the continuous
extension to $[0,1]^2$ of
\begin{equation}
\label{eq1: binary RD}
d_\alpha (p \| q ) =\frac1{\alpha-1} \; \log \Bigl(p^{\alpha} q^{1-\alpha}
+ (1-p)^{\alpha} (1-q)^{1-\alpha} \Bigr).
\end{equation}
\end{definition}

Used several times in this paper, it is easy to verify the following
identity satisfied by the binary R\'{e}nyi divergence. If $t\in [0,1]$
and $s \in [0, \theta]$ for $\theta >0$, then
\begin{align} \label{useful equality}
\log \theta - d_{\alpha}\bigl(t \| \tfrac{s}{\theta}\bigr)
= \frac1{1-\alpha} \; \log \left( t^\alpha \, s^{1-\alpha}
+ (1-t)^\alpha \, (\theta-s)^{1-\alpha} \right).
\end{align}

By analytic continuation in $\alpha$, for $(p, q) \in (0,1)^2$,
\begin{align}
\label{d_0}
d_0(p\|q) &= 0, \\[0.1cm]
d_1(p\|q) &= d(p\|q) =
p \log \frac{p}{q} + (1-p) \log \frac{1-p}{1-q}, \label{d_1} \\[0.1cm]
\label{d_infty}
d_{\infty}(p \| q) &= \log \max \left\{\frac{p}{q}, \,
\frac{1-p}{1-q} \right\}
\end{align}
where $d(\cdot \| \cdot)$ denotes the binary relative entropy.
Note that if
$t\in [0,1]$ and $M>1$, then
\begin{align}\label{good4fano}
\log M - d ( t \,\| 1 - \tfrac1M) = t\, \log (M-1) + h(t).
\end{align}

A simple nexus between the R\'enyi entropy and the R\'enyi divergence is
\begin{align}
\label{eq: RD - U is uniform}
D_{\alpha}(X \| U) = \log M - H_\alpha(X)
\end{align}
when $X$ takes values on a set of $M$ elements on which $U$ is equiprobable.

\begin{definition}\label{def:chernoff} The \textit{Chernoff
information} of a pair of probability measures defined on
the same measurable space is equal to
\begin{align} \label{eq:def:chernoffinformation}
C ( P \| Q ) = \sup_{\alpha \in (0,1)} (1 - \alpha ) D_\alpha ( P \| Q ).
\end{align}
\end{definition}

\begin{definition} (\cite{Sibson69}, \cite{Verdu-ITA15})
Given the probability distributions $P_X$ and $P_{Y|X}$, the $\alpha$-\textit{mutual information}
is defined for $\alpha > 0$ as
\begin{align}  \label{def: PV - alpha MI}
& I_{\alpha}(X; Y) = \min_{Q_Y}  D_{\alpha}(P_{XY} \| P_X \times Q_Y).
\end{align}
As for the conventional mutual information, which refers to $\alpha=1$, sometimes it is useful to employ the
alternative notation
\begin{align}
I_\alpha (P_X , P_{Y|X} ) = I_{\alpha}(X; Y).
\end{align}
\end{definition}

In the discrete case, for $\alpha \in (0,1) \cup (1, \infty)$,
\begin{align}  \label{eq: E0_a}
I_{\alpha}(X; Y) &= \frac{\alpha}{1 -\alpha}  \; E_0 \left( \frac1{\alpha}-1, P_X\right),
\end{align}
where
\begin{align}
& E_0( \rho, P_X) =
-\log \sum_{y \in \set{Y}} \left( \sum_{x \in \set{X}} P_X(x)
\, P_{Y|X}^{\frac1{1+\rho}}(y|x) \right)^{1+\rho}  \label{eq0: E0_b}
\end{align}
denotes Gallager's error exponent function \cite{Gallager_IT1965}.
Note that, in general, $I_\alpha (X;Y) $ does not correspond to the
difference $H_\alpha (X) - H_\alpha (X|Y)$, which is equal to
\begin{align}
H_\alpha (X) - H_\alpha (X|Y) = \frac{\alpha}{1-\alpha}  \;
E_0 \left( \frac1{\alpha}-1, P_{X_{\alpha}}\right) \label{eq: E0_b}
\end{align}
where
$P_{X_{\alpha}}$ is the scaled version of $P_X$,  namely the normalized
version of $P_X^{\alpha}$. Since the equiprobable distribution is equal
to its scaled version for any $\alpha \geq 0$, if $X$ is equiprobable
on a set of $M$ elements, then
\begin{align}\label{equiprobablerenyialpha}
I_{\alpha}(X; Y) &= \log M - H_\alpha (X|Y).
\end{align}

\section{Upper and Lower Bounds on $H_{\alpha}(X)$}
\label{section: tight bound and joint range of RE}

In this section, we obtain upper and lower bounds on the unconditional
R\'enyi entropy of order $\alpha$ which are tight in the sense that
they are attained with equality for certain random variables.

\subsection{Upper bounds}
In this subsection we limit ourselves to finite alphabets.
\begin{theorem}  \label{theorem: generalized UB of RE}
Let $\set{X}$ and $\set{Y}$ be finite sets, and let $X$ be
a random variable taking values on $\set{X}$.
Let $f\colon \set{X} \to \set{Y}$ be a deterministic function,
and denote the cardinality of the inverse image by
$L_y = | f^{-1}(y)|$ for every $y \in \set{Y}$.
Then, for every $\alpha \in [0,1) \cup (1, \infty)$,
the R\'{e}nyi entropy of $X$ satisfies
\begin{align}  \label{eq: generalized UB on RE}
H_\alpha (X) \leq \frac{1}{1-\alpha} \, \log
\sum_{y \in \set{Y}} L_y^{1-\alpha} \;
\prob^\alpha[f(X)= y]
\end{align}
which holds with equality if and only if either $\alpha =0$ or
$\alpha \in (0,1) \cup (1, \infty)$ and $P_X$ is equiprobable
on $f^{-1}(y)$ for every $y \in \set{Y}$ such that $P_X(f^{-1}(y)) > 0$.
\end{theorem}

\begin{proof}
It can be verified from \eqref{eq: RE of zero order} that
\eqref{eq: generalized UB on RE} holds with equality for
$\alpha=0$ (by assumption $ \supp P_X = \set{X}$, and
$\underset{y \in \set{Y}}{\sum} L_y = |\set{X}|$).

Suppose $\alpha \in (0,1) \cup (1, \infty)$. Let $|\set{X}|=M<\infty$,
and let $U$ be equiprobable on $\set{X}$. Let $V = f(X)$ and
$W = f(U)$, so $(V, W) \in \set{Y}^2$ and
\begin{align}
\label{eq1: V,W}
& P_V(y) = \prob \bigl[ f(X) = y \bigr], \\
\label{eq2: V,W}
& P_W(y) = \frac{L_y}{M}
\end{align}
for all $y \in \set{Y}$.
To show \eqref{eq: generalized UB on RE}, note that
\begin{align}
H_{\alpha}(X) &= \log M - D_{\alpha}(X \| U) \label{eq:th1-1} \\
& \leq \log M - D_{\alpha}(V \| W) \label{eq:th1-2} \\
& = \log M - \frac1{\alpha-1} \, \log  \sum_{y \in \set{Y}}
\prob^{\alpha} \bigl[ f(X) = y \bigr]
\left(\frac{L_y}{M}\right)^{1-\alpha}  \label{eq:th1-4} \\
& = \frac{1}{1-\alpha} \, \log
\sum_{y \in \set{Y}} L_y^{1-\alpha} \;
\prob^\alpha[f(X)= y]  \label{eq:th1-5}
\end{align}
where \eqref{eq:th1-1} is \eqref{eq: RD - U is uniform};
\eqref{eq:th1-2} holds due to the data processing inequality
for the R\'{e}nyi divergence
(see \cite[Theorem~9 and Example~2]{ErvenH14});
and \eqref{eq:th1-4} follows from \eqref{eq:RD1},
\eqref{eq1: V,W} and \eqref{eq2: V,W}.

From \eqref{eq:th1-1}--\eqref{eq:th1-5}, the upper bound in
\eqref{eq: generalized UB on RE} is attained with equality if and only
if the data processing inequality \eqref{eq:th1-2} holds with equality.
For all $\alpha \in (0,1) \cup (1, \infty)$,
\begin{align}
D_{\alpha}(X \| U)
\label{eq:th1-6}
& = \frac1{\alpha-1} \, \log \sum_{x \in \set{X}}
P_X^{\alpha}(x) M^{\alpha-1}  \\[0.1cm]
\label{eq:th1-7}
& = \frac1{\alpha-1} \, \log \left( \sum_{y \in \set{Y}} L_y
M^{\alpha-1}
\sum_{x \in f^{-1}(y)} \frac1{L_y} \, P_X^{\alpha}(x) \right)  \\[0.1cm]
\label{eq:th1-8}
& \geq \frac1{\alpha-1} \, \log \left( \sum_{y \in \set{Y}} L_y
M^{\alpha-1} \left( \frac1{L_y} \sum_{x \in f^{-1}(y)}
\, P_X(x) \right)^{\alpha} \, \right) \\[0.1cm]
\label{eq:th1-9}
& =  \frac1{\alpha-1} \, \log \sum_{y \in \set{Y}} \,
P_V^{\alpha}(y) \left(\frac{L_y}{M}\right)^{1-\alpha} \\
\label{eq:th1-11}
& = D_{\alpha}(V \| W)
\end{align}
where \eqref{eq:th1-6} is satisfied by the definition in \eqref{eq:RD1}
with $U$ being equiprobable on $\set{X}$ and $|\set{X}|=M$;
\eqref{eq:th1-7} holds by expressing the sum over $\set{X}$ as a double sum
over $\set{Y}$ and the elements of $\set{X}$ that are mapped by $f$ to
the same element in $\set{Y}$; inequality \eqref{eq:th1-8}
holds since $|f^{-1}(y)| = L_y$ for all $y \in \set{Y}$, and for every
random variable $Z$ we have $\mathbb{E}[Z^\alpha] \geq \mathbb{E}^{\alpha}[Z]$
if $\alpha \in (1, \infty)$ or the opposite inequality if $\alpha \in [0,1)$;
finally, \eqref{eq:th1-9}--\eqref{eq:th1-11} are due to \eqref{eq:RD1},
\eqref{eq1: V,W} and \eqref{eq2: V,W}. Note that \eqref{eq:th1-8} holds
with equality if and only if $P_X$ is equiprobable
on $f^{-1}(y)$ for every $y \in \set{Y}$ such that $P_X(f^{-1}(y)) > 0$.
Hence, for $\alpha \in (0,1) \cup (1, \infty)$, \eqref{eq:th1-2} holds
with equality if and only if the condition given in the theorem statement
is satisfied.
\end{proof}

\begin{corollary} \label{corollary: generalized UB of RE}
Let $X$ be a random variable taking values on a finite set
$\set{X}$ with $|\set{X}|=M$, and let $\set{L} \subseteq \set{X}$
with $|\set{L}|=L$.
Then, for all $\alpha \in [0,1) \cup (1, \infty)$,
\begin{align}
\label{eq: gen. UB on RE}
H_{\alpha}(X) & \leq
\frac1{1-\alpha} \; \log \Bigl( L^{1-\alpha} \,
\prob^{\alpha}[X \in \set{L}] + (M-L)^{1-\alpha} \,
\prob^{\alpha}[X \notin \set{L}] \Bigr) \\[0.1cm]
\label{eq2: gen. UB on RE}
& = \log M - d_{\alpha} \left( \prob[X \in \set{L}] \| \tfrac{L}{M} \right)
\end{align}
with equality in \eqref{eq: gen. UB on RE} if and only if
$X$ is equiprobable on both $\set{L}$ and $\set{L}^{\mathrm{c}}$.
\end{corollary}
\begin{proof}
Inequality~\eqref{eq: gen. UB on RE} is a specialization of
Theorem~\ref{theorem: generalized UB of RE} to a binary-valued
function $f \colon \set{X} \to \{0,1\}$
such that $f^{-1}(0) = \set{L}$ and $f^{-1}(1) = \set{L}^{\mathrm{c}}$.
Equality \eqref{eq2: gen. UB on RE} holds by setting $\theta = M$,
$t = \prob[X \in \set{L}]$, and $s=L$ in \eqref{useful equality}.
\end{proof}

Moreover, specializing Corollary~\ref{corollary: generalized UB of RE} to
$L=1$, we obtain:
\begin{corollary} \label{corollary: UB of RE}
Let $X$ be a random variable taking values on a finite set $\set{X}$
with $|\set{X}|=M$. Then, for all $\alpha \in [0,1) \cup (1, \infty)$,
\begin{align}  \label{eq: UB on RE}
H_{\alpha}(X) & \leq \min_{x \in \set{X}}
\, \frac1{1-\alpha} \, \log \Bigl( P_X^{\alpha}(x) + (M-1)^{1-\alpha} \,
\bigl(1-P_X(x)\bigr)^{\alpha} \Bigr) \\[0.1cm]
& = \log M - \max_{x \in \set{X}} \, d_{\alpha}\bigl( P_X(x) \| \tfrac1M \bigr)
\label{oneof}
\end{align}
with equality in \eqref{eq: UB on RE} if and only if $P_X$ is equiprobable
on $\set{X} \setminus \{x^\ast\}$ where $x^\ast \in \set{X}$ attains the
maximum in \eqref{oneof}.
\end{corollary}

\begin{remark}
Taking the limit $\alpha \to 1$ in the right side of
\eqref{eq: UB on RE} yields
\begin{align}  \label{eq: UB on H}
H(X) \leq \min_{x \in \set{X}} \Bigl\{ \bigl(1-P_X(x)\bigr)
\, \log(M-1) + h\bigl( P_X(x) \bigr) \Bigr\}
\end{align}
with the same necessary and sufficient condition for
equality in Corollary~\ref{corollary: UB of RE}. Furthermore,
choosing $x \in \set{X}$ to be a mode of $P_X$ gives
\begin{align}
\label{eq2: UB on H}
H(X) & \leq \bigl( 1-p_{\max} \bigr)
\, \log(M-1) + h\bigl( p_{\max} \bigr) \\[0.05cm]
\label{eq3: UB on H}
& = \varepsilon_X \, \log(M-1) + h(\varepsilon_X)
\end{align}
where \eqref{eq3: UB on H} follows from \eqref{eq: blind decision},
and the symmetry of the binary entropy function around $\tfrac12$.
\end{remark}

If we loosen Corollary~\ref{corollary: UB of RE} by, instead of
minimizing the right side of \eqref{eq: UB on RE}, choosing $x \in \set{X}$
to be a mode of $P_X$, we recover \cite[Theorem~6]{Ben-Bassat-Raviv}:
\begin{corollary}  \label{corollary: UB2 of RE}
Let $X$ be a random variable taking $M$ possible values,
and assume that its largest mass is $p_{\max}$. Then, for all
$\alpha \in [0,1) \cup (1, \infty)$,
\begin{align}  \label{eq: specialized UB on RE}
H_{\alpha}(X) & \leq \frac1{1-\alpha} \; \log \Bigl( p_{\max}^{\alpha}
+ (M-1)^{1-\alpha} \, \bigl(1-p_{\max}\bigr)^{\alpha} \Bigr) \\[0.1cm]
\label{eq2: specialized UB}
& = \log M - d_{\alpha}\bigl( p_{\max} \| \tfrac1M \bigr)
\end{align}
with equality in \eqref{eq: specialized UB on RE} under the
condition in the statement of Corollary~\ref{corollary: UB of RE}
with $x^\ast$ being equal to a mode of $P_X$.
\end{corollary}

\begin{remark} \label{remark: minimization over the set}
In view of the necessary and sufficient condition for
equality in \eqref{eq: UB on RE}, the minimization of
the upper bound in the right side of \eqref{eq: UB on RE}
for a given $P_X$ does not necessarily imply that the best
choice of $x \in \set{X}$ is a mode of $P_X$. For example, let
$\set{X} = \{0, 1, 2\}$, $P_X(0) = 0.2$, and $P_X(1) = P_X(2) = 0.4$.
Then, $H_2(X) = 1.474$ bits, which coincides with its upper bound
in \eqref{eq: UB on RE} by choosing $x=0$; on the other hand,
\eqref{eq: specialized UB on RE} yields $H_2(X) \leq 1.556$~bits.
\end{remark}

\subsection{Lower bound}
The following lower bound provides a counterpart to the upper
bound in Corollary~\ref{corollary: UB2 of RE} without restricting
to finitely valued random variables.
\begin{theorem} \label{theorem: LB of RE}
Let $X$ be a discrete random variable attaining a maximal mass $p_{\max}$.
Then, for $\alpha \in (0,1) \cup (1, \infty)$,
\begin{align} \label{eq: LB on RE of order alpha}
H_{\alpha}(X) & \geq \frac1{1-\alpha} \; \log \left( \bigg\lfloor
\frac1{p_{\max}} \bigg\rfloor \, p_{\max}^{\alpha} + \left(1 - p_{\max}
\bigg\lfloor \frac1{p_{\max}} \bigg\rfloor \right)^{\alpha} \right) \\[0.1cm]
\label{eq2: LB on RE of order alpha}
& = \log \left(1 + \bigg\lfloor \frac1{p_{\max}} \bigg\rfloor \right)
- d_{\alpha} \left( p_{\max} \bigg \lfloor \frac1{p_{\max}} \bigg \rfloor
\, \bigg \| \, \frac{\Big \lfloor \frac1{p_{\max}} \Big\rfloor}{1+
\Big \lfloor \frac1{p_{\max}} \Big\rfloor}\right)
\end{align}
where, for $x \in \Reals$, $\lfloor x \rfloor$ denotes the largest integer
that is smaller than or equal to $x$.
Equality in \eqref{eq: LB on RE of order alpha} holds if and only
if $P_X$ has $\Big\lfloor \tfrac1{p_{\max}} \Big\rfloor$ masses equal
to $p_{\max}$, and an additional mass equal to
$1 - p_{\max} \, \Big\lfloor \frac1{p_{\max}} \Big\rfloor$
whenever $\frac1{p_{\max}}$ is not an integer.
\end{theorem}

\begin{proof}
For $\alpha \in (0,1) \cup (1, \infty)$, the R\'{e}nyi entropy of the
distribution identified in the statement of Theorem~\ref{theorem: LB of RE}
as attaining the condition for equality is equal to the right side of
\eqref{eq: LB on RE of order alpha}. Furthermore, that distribution
majorizes any distribution whose maximum mass is $p_{\max}$. The result
in \eqref{eq: LB on RE of order alpha} follows from the Schur-concavity
of the R\'{e}nyi entropy in the general case of a countable alphabet
(see \cite[Theorem~2]{HoS-ISIT15}). In view of Lemma~1 and Theorem~2
of \cite{HoS-ISIT15}, the Schur concavity of the R\'enyi entropy is
strict for any $\alpha \in (0,1) \cup (1, \infty)$ and therefore
\eqref{eq: LB on RE of order alpha} holds with strict inequality
for any distribution other than the one specified in the statement
of this result. To get \eqref{eq2: LB on RE of order alpha}, let
$s = \Bigl \lfloor \frac1{p_{\max}} \Bigr \rfloor$,
$t = s\,  p_{\max}$ and $\theta = s+1$ in \eqref{useful equality}.
\end{proof}

\begin{remark}
Let $X$ be a discrete random variable, and consider the case
of $\alpha = 0$. Unless $X$ is finitely valued,
\eqref{eq: RE of zero order} yields $H_0(X) = \infty$.
If $X$ is finitely valued, then \eqref{eq: RE of zero order} and
the inequality $p_{\max} \, | \supp P_X | \geq 1$ yield
\begin{align} \label{20170904-2}
H_0(X) \geq \log \left\lceil \tfrac1{p_{\max}} \right\rceil
\end{align}
where, for $x \in \Reals$, $\lceil x \rceil$ denotes the smallest integer
that is larger than or equal to $x$. The bound in \eqref{eq: LB on RE of order alpha}
therefore holds for $\alpha=0$ with the convention that $0^0 = 0$, and it then
coincides with \eqref{20170904-2}. Equality in \eqref{20170904-2} holds if and
only if $|\supp P_X| = \left\lceil \tfrac1{p_{\max}} \right\rceil$
(e.g., if $X$ is equiprobable).
\end{remark}

\begin{remark}
Taking the limit $\alpha \to 1$ in Theorem~\ref{theorem: LB of RE} yields
\begin{align} \label{eq: LB on H}
H(X) & \geq h\left(p_{\max} \bigg\lfloor \frac1{p_{\max}}
\bigg \rfloor \right) + p_{\max} \bigg\lfloor \frac1{p_{\max}}
\bigg\rfloor \, \log \bigg\lfloor \frac1{p_{\max}} \bigg\rfloor \\
\label{eq2: LB on H}
& = \log \left(1 + \bigg\lfloor \frac1{p_{\max}} \bigg\rfloor \right)
- d \left( p_{\max} \bigg \lfloor \frac1{p_{\max}} \bigg \rfloor
\, \bigg \| \, \frac{\Big \lfloor \frac1{p_{\max}} \Big\rfloor}{1+
\Big \lfloor \frac1{p_{\max}} \Big\rfloor}\right)
\end{align}
with equality in \eqref{eq: LB on H} if and only if the condition
for equality in Theorem~\ref{theorem: LB of RE} holds. Hence,
the result in Theorem~\ref{theorem: LB of RE} generalizes the bound
in \eqref{eq: LB on H}, which is due to Kovalevsky \cite{Kovalevsky}
and Tebbe and Dwyer \cite{TebbeD68} in the special case of a
finitely-valued $X$ (rediscovered in \cite{FederM94}), and to Ho and
Verd\'u \cite[Theorem~10]{HoS-IT10} in the general setting of a
countable alphabet.
\end{remark}

\section{Arimoto-R\'{e}nyi Conditional Entropy and Error Probability}
\label{section: cond. RE vs. Pe}
Section~\ref{section: cond. RE vs. Pe} forms the main part of this paper,
and its results are outlined in Section~\ref{subsec: 1.4}
(see Items~\ref{s:1})--\ref{s:5})).

\subsection{Upper bound on the Arimoto-R\'{e}nyi conditional entropy: Generalized Fano's inequality}
\label{subsection: Generalized Fano inequality}

The minimum error probability  $\varepsilon_{X|Y}$
can be achieved by a deterministic function
(\textit{maximum-a-posteriori} decision rule)
$\mathcal{L}^\ast \colon \set{Y} \to \set{X}$:
\begin{align}
\varepsilon_{X|Y}
&= \min_{\mathcal{L} \colon \set{Y} \to \set{X}}
\mathbb{P} [ X \neq \mathcal{L} (Y) ] \label{20170904} \\
&= \mathbb{P} [ X \neq \mathcal{L}^\ast (Y) ] \label{eq:MAP}\\
&= 1- \mathbb{E} \left[ \max_{x \in \set{X}}
P_{X|Y}(x|Y) \right]
\label{eq1: cond. epsilon} \\
&\leq 1 - p_{\max} \label{USA}\\
&\leq 1 - \frac1M \label{UK}
\end{align}
where \eqref{USA} is the minimum error probability
achievable among blind decision rules that disregard
the observations (see \eqref{eq: blind decision}).
\par
Fano's inequality links the decision-theoretic uncertainty
$\varepsilon_{X|Y}$ and the information-theoretic uncertainty
$H(X|Y)$
through
\begin{align} \label{eq: Fano}
H(X | Y) & \leq \log M - d\bigl( \varepsilon_{X|Y} \| 1-\tfrac1M \bigr) \\
\label{eq2: Fano}
& = h(\varepsilon_{X|Y}) + \varepsilon_{X|Y} \log(M-1)
\end{align}
where the identity in \eqref{eq2: Fano} is \eqref{good4fano} with $t = \varepsilon_{X|Y}$.
Although the form of Fano's inequality in \eqref{eq: Fano}
is not nearly as popular as \eqref{eq2: Fano}, it turns out to
be the version that admits an elegant (although not immediate) generalization
to the Arimoto-R\'{e}nyi conditional entropy. It is straightforward to obtain
\eqref{eq2: Fano} by averaging a conditional version of \eqref{eq3: UB on H}
with respect to the observation. This simple route to the desired
result is not viable in the case of $H_\alpha (X|Y)$ since it
is not an average of R\'enyi entropies of conditional distributions.
The conventional proof of Fano's inequality (e.g., \cite[pp.~38--39]{Cover_Thomas}),
based on the use of the chain rule for  entropy, is also doomed to failure for
the Arimoto-R\'{e}nyi conditional entropy of order $\alpha \neq 1$ since it
does not satisfy the chain rule.

Before we generalize Fano's inequality by linking $\varepsilon_{X|Y}$
with $H_\alpha (X|Y)$ for $\alpha \in (0, \infty)$, note that for
$\alpha = \infty$, the following generalization of \eqref{eq2: min error prob.}
holds  in view of \eqref{eq: cond. RE at infinity} and
\eqref{eq1: cond. epsilon}:
\begin{align} \label{varepsilon_X|Y}
\varepsilon_{X|Y} = 1 - \exp\bigl(-H_{\infty}(X | Y) \bigr).
\end{align}

\begin{theorem}
\label{theorem: generalized Fano inequality}
Let $P_{XY}$ be a probability measure defined on
$\set{X} \times \set{Y}$ with $|\set{X}|=M < \infty$.
For all $\alpha \in (0, \infty)$,
\begin{align} \label{eq1: generalized Fano}
H_{\alpha}(X | Y) \leq \log M -
d_{\alpha}\bigl( \varepsilon_{X|Y} \| 1-\tfrac1M \bigr).
\end{align}
Let $\mathcal{L}^\ast \colon \set{Y} \to \set{X}$ be a
deterministic MAP decision rule. Equality holds in
\eqref{eq1: generalized Fano} if and only if, for
all $y \in \set{Y}$,
\begin{align} \label{eq: equality in Fano-nolist}
P_{X|Y}(x|y) =
\begin{dcases}
\frac{\varepsilon_{X|Y}}{M-1}, & \quad x \neq \mathcal{L}^\ast(y), \\[0.2cm]
1-\varepsilon_{X|Y}, & \quad x=\mathcal{L}^\ast(y).
\end{dcases}
\end{align}
\end{theorem}

\begin{proof}
If in Corollary~\ref{corollary: UB2 of RE} we replace $P_X$ by $P_{X|Y=y}$, then
$p_{\max}$ is replaced by $\underset{x \in \set{X}}{\max} \, P_{X|Y}(x|y)$ and we obtain
\begin{align} \label{th3-0}
H_{\alpha}(X | Y=y) \leq \frac1{1-\alpha} \;
\log \Bigl( \bigl(1-\varepsilon_{X|Y}(y)\bigr)^{\alpha} + (M-1)^{1-\alpha} \,
\varepsilon^{\alpha}_{X|Y}(y) \Bigr)
\end{align}
where we have defined the conditional error probability given the observation:
\begin{align} \label{Aachen}
\varepsilon_{X|Y}(y)= 1-  \max_{x \in \set{X}} P_{X|Y}(x|y) ,
\end{align}
which satisfies with $Y \sim P_Y$,
\begin{align}\label{eq1: avg. cond. epsilon}
\varepsilon_{X|Y} = \mathbb{E} [ \varepsilon_{X|Y}(Y) ].
\end{align}
For $\alpha \in (0,1) \cup (1, \infty)$, $(\beta, \gamma) \in (0, \infty)^2$,
define the function $f_{\alpha, \beta, \gamma} \colon [0,1] \to (0, \infty)$:
\begin{align} \label{eq: f}
f_{\alpha, \beta, \gamma}(u) = \left( \gamma (1-u)^{\alpha} + \beta u^{\alpha}
\right)^{\frac1\alpha}.
\end{align}
If $\alpha > 1$, then \eqref{eq2: Arimoto - cond. RE} allows us to bound a
monotonically decreasing function of $H_{\alpha}(X | Y)$:
\begin{align}
\exp \left( \frac{1-\alpha}{\alpha} \; H_{\alpha}(X | Y) \right)
\label{th3-2}
& \geq \mathbb{E} \left[ \left(
\bigl(1-\varepsilon_{X|Y}(Y)\bigr)^{\alpha} + \beta \,
\varepsilon_{X|Y}^{\alpha}(Y) \right)^{\frac1\alpha} \right] \\[0.2cm]
\label{th3-3}
& = \mathbb{E} \left[ f_{\alpha, \beta, 1}
\bigl(\varepsilon_{X|Y}(Y)\bigr) \right] \\[0.1cm]
\label{th3-4}
& \geq f_{\alpha, \beta, 1} \bigl( \varepsilon_{X|Y} \bigr) \\
\label{th3-5}
& = \Bigl( \bigl(1-\varepsilon_{X|Y}\bigr)^{\alpha} + (M-1)^{1-\alpha} \,
\varepsilon_{X|Y}^{\alpha} \Bigr)^{\frac1\alpha}
\end{align}
where $\beta = (M-1)^{1-\alpha}$ in \eqref{th3-2}--\eqref{th3-4};
\eqref{th3-2} follows from \eqref{eq2: Arimoto - cond. RE} and \eqref{th3-0};
\eqref{th3-4} follows from \eqref{eq1: avg. cond. epsilon} and Jensen's
inequality, due to the convexity of $f_{\alpha, \beta, \gamma} \colon [0,1] \to (0, \infty)$ for
$\alpha \in (1, \infty)$ and $(\beta, \gamma) \in (0, \infty)^2$
(see Lemma~\ref{lemma: convexity/ concavity} following this proof);
finally, \eqref{th3-5} follows from \eqref{eq: f}.
For $\alpha \in (1, \infty)$, \eqref{eq1: generalized Fano} follows from
\eqref{th3-2}--\eqref{th3-5} and identity \eqref{useful equality}
with $(\theta, s, t) = (M, \, M-1, \, \varepsilon_{X|Y})$.

For $\alpha \in (0,1)$, inequality~\eqref{th3-2} is reversed; and due to the
concavity of $f_{\alpha, \beta, \gamma}(\cdot)$ on $[0,1]$
(Lemma~\ref{lemma: convexity/ concavity}),
inequality~\eqref{th3-4} is also reversed. The rest of the proof proceeds as in
the case where $\alpha > 1$.

The necessary and sufficient condition for equality in \eqref{eq1: generalized Fano}
follows from the condition for equality in the statement of Corollary~\ref{corollary: UB2 of RE}
when $P_X$ is replaced by $P_{X|Y=y}$ for $y \in \set{Y}$. Under
\eqref{eq: equality in Fano-nolist}, it follows from \eqref{20170904}--\eqref{UK}
and \eqref{Aachen} that $\varepsilon_{X|Y}(y) = \varepsilon_{X|Y}$ for all
$y \in \set{Y}$; this implies that inequalities \eqref{th3-2} and \eqref{th3-4}
for $\alpha \in (1, \infty)$, and the opposite inequalities for $\alpha \in (0,1)$,
hold with equalities.
\end{proof}

\begin{lemma} \label{lemma: convexity/ concavity}
Let $\alpha \in (0,1) \cup (1, \infty)$ and $(\beta, \gamma) \in (0, \infty)^2$.
The function $f_{\alpha, \beta, \gamma} \colon [0,1] \to (0, \infty)$ defined
in \eqref{eq: f} is strictly convex  for $\alpha \in (1, \infty)$, and strictly
concave for $\alpha \in (0,1)$.
\end{lemma}
\begin{proof}
The second derivative of $f_{\alpha, \beta, \gamma}(\cdot)$ in \eqref{eq: f}
is given by
\begin{align}
f''_{\alpha, \beta, \gamma}(u) = (\alpha-1) \beta \gamma
\Bigl( \gamma (1-u)^{\alpha} + \beta u^{\alpha} \Bigr)^{\frac1\alpha-2}
\bigl( u(1-u) \bigr)^{\alpha-2}
\end{align}
which is strictly negative if $\alpha \in (0,1)$, and strictly positive if
$\alpha \in (1, \infty)$ for any $u \in [0,1]$.
\end{proof}

\begin{remark}
From \eqref{d_1}, Fano's inequality (see \eqref{eq: Fano}--\eqref{eq2: Fano}) is
recovered by taking the limit $\alpha \to 1$ in \eqref{eq1: generalized Fano}.
\end{remark}

\begin{remark}
A pleasing feature of Theorem~\ref{theorem: generalized Fano inequality}
is that as $\alpha \to \infty$, the bound becomes tight.
To see this, we rewrite the identity in \eqref{varepsilon_X|Y} as
\begin{align}
H_{\infty}(X | Y)
\label{eq: haifa2}
&= \log M - \log \frac{1-\varepsilon_{X|Y}}{\tfrac1M} \\[0.1cm]
\label{eq: haifa3}
&= \log M - d_{\infty}\bigl( \varepsilon_{X|Y} \| 1-\tfrac1M \bigr)
\end{align}
where \eqref{eq: haifa3} follows from \eqref{d_infty} since
$\varepsilon_{X|Y} \leq 1 - \tfrac1M$ (see \eqref{20170904}--\eqref{UK}).
\end{remark}
\begin{remark} For $\alpha = 0$, \eqref{eq1: generalized Fano} also holds.
To see this, it is useful to distinguish two cases:
\begin{itemize}
\item
$\varepsilon_{X|Y} = 0$. Then, $H_0 (X | Y ) = 0$, and the right side of
\eqref{eq1: generalized Fano} is also equal to zero since $d_0(0 \| q ) = -\log(1 - q)$
for all $q \in [0,1)$.
\item
$\varepsilon_{X|Y} > 0$. Then, the right side of \eqref{eq1: generalized Fano}
is equal to $\log M$ (see \eqref{d_0}) which is indeed an upper bound to $H_0 (X|Y)$.
The condition for equality in this case is that there exists $y \in \set{Y}$
such that $P_{X|Y}( x|y) >0$ for all $x \in \set{X}$.
\end{itemize}
\end{remark}

\begin{remark}
Since $d_\alpha \bigl( \cdot \| 1-\tfrac1M \bigr) $ is monotonically decreasing
in $[0, 1-\tfrac1M]$, Theorem~\ref{theorem: generalized Fano inequality}
gives an implicit lower bound on $\varepsilon_{X|Y}$. Although, currently, there
is no counterpart to Shannon's explicit lower bound as a function of $H(X|Y)$
\cite{Shannon58}, we do have explicit lower bounds as a function of $H_\alpha(X|Y)$
in Section~\ref{subsec: explicit LBs}.
\end{remark}

\begin{remark}
If $X$ and $Y$ are random
variables taking values on a set of $M$ elements and $X$ is equiprobable,
then  \cite[Theorem~5.3]{PolyanskiyV10} shows that
\begin{align}  \label{eq: PV - alpha MI}
& I_{\alpha}(X; Y) \geq d_{\alpha} ( \varepsilon_{X|Y} \| 1-\tfrac1M),
\end{align}
which, together with \eqref{equiprobablerenyialpha}, yields \eqref{eq1: generalized Fano}.
However, note that in Theorem~\ref{theorem: generalized Fano inequality} we do not restrict
$X$ to be equiprobable.
\end{remark}

In information-theoretic problems, it is common to encounter the case in which
$X$ and $Y$ are actually vectors of dimension $n$. Fano's inequality ensures
that vanishing error probability implies vanishing normalized conditional
entropy as $n\to \infty$. As we see next, the picture with the Arimoto-R\'enyi
conditional entropy is more nuanced.
\begin{theorem}\label{thm:genhov}
Let $\{X_n\}$ be a sequence of random variables, with $X_n$ taking values
on $\set{X}_n$ for $n \in \naturals$, and assume that there exists
an integer $M \geq 2$ such that $|\set{X}_n| \leq M^n$ for all $n$.\footnote{Note
that this encompasses the conventional setting in which $\set{X}_n = \set{A}^n$.}
Let $\{Y_n\}$ be an arbitrary sequence of random variables, for which
$\varepsilon_{X_n | Y_n} \to 0$ as $n \to \infty$.
\begin{enumerate}[a)]
\item\label{thm:genhov:a}  If $\alpha \in (1, \infty]$, then $H_{\alpha}(X_n | Y_n) \to 0$;
\item\label{thm:genhov:a2} If $\alpha=1$, then $\frac1n \, H(X_n | Y_n) \to 0$;
\item\label{thm:genhov:b}  If $\alpha \in [0,1)$, then $\frac1{n} \, H_{\alpha}(X_n | Y_n)$
is upper bounded by $\log M$; nevertheless, it does not necessarily tend to 0.
\end{enumerate}
\end{theorem}
\begin{proof}
\begin{enumerate}[a)]
\item
For $\alpha \in (1, \infty)$,
\begin{align}
H_{\alpha}(X_n | Y_n) & \leq
n \log M - d_{\alpha}\bigl( \varepsilon_{X_n | Y_n} \, \| \, 1-M^{-n} \bigr) \label{z0} \\
&= \frac1{1-\alpha} \, \log \left( \varepsilon_{X^n|Y^n}^{\alpha} (M^n-1)^{1-\alpha} +
\bigl(1-\varepsilon_{X_n | Y_n}\bigr)^{\alpha} \right), \label{z1}
\end{align}
where \eqref{z0} follows from \eqref{eq1: generalized Fano} and $|\set{X}_n| \leq M^n$;
\eqref{z1} holds due to \eqref{useful equality} by setting the parameters $\theta = M^n$,
$s = M^n-1$ and $t = \varepsilon_{X^n | Y^n}$;
hence, $H_{\alpha}(X_n | Y_n) \to 0$ since $\alpha > 1$ and $\varepsilon_{X_n | Y_n} \to 0$.
Item~\ref{thm:genhov:a})
also holds for $\alpha = \infty$ since $H_{\alpha}(X_n | Y_n)$ is monotonically decreasing in
$\alpha$ throughout the real line (Proposition~\ref{prop1: mon}).
\item
For $\alpha=1$, from Fano's inequality,
\begin{align}
\tfrac1n \, H(X_n | Y_n) &\leq
\varepsilon_{X_n | Y_n} \, \log M + \tfrac1n \, h\bigl(\varepsilon_{X_n | Y_n}\bigr) \to 0.
\end{align}
Hence, not only does $\tfrac1n \, H(X_n | Y_n) \to 0$ if $\varepsilon_{X_n | Y_n} = o(1)$
but also $ H(X_n | Y_n) \to 0$ if $\varepsilon_{X_n | Y_n} = o\left( \frac1n \right)$.
\item Proposition~\ref{prop1: mon} implies that if $\alpha > 0$, then
\begin{align}
H_{\alpha}(X_n | Y_n) &\leq H_\alpha(X_n) \\
&\leq \log |\set{X}_n| \\
&\leq n \log M.
\end{align}
\end{enumerate}
\end{proof}
A counterexample where $\varepsilon_{X_n | Y_n} \to 0$ exponentially fast, and yet
$\frac1{n} H_{\alpha}(X_n | Y_n) \not \to 0$ for all $\alpha \in [0,1)$ is given as follows.
\begin{remark}\label{noto}
In contrast to the conventional case of $\alpha=1$, it is not possible to strengthen
Theorem~\ref{thm:genhov}\ref{thm:genhov:b}) to claim that $\varepsilon_{X_n | Y_n} \to 0$
implies that $ \frac1n H_{\alpha}(X_n | Y_n) \to 0$ for $\alpha \in [0,1).$
By Proposition~\ref{prop1: mon}, it is sufficient to consider the following counterexample:
fix any $\alpha \in (0,1)$, and let $Y_n$ be deterministic, $\set{X}_n = \{1, \ldots, M^n\}$,
and $X_n \sim P_{X_n}$ where
\begin{align} \label{eq: dist. X_n}
P_{X_n} = \left[1 - \beta^{-n}, ~ \frac{\beta^{-n}}{M^n-1}, \ldots, \frac{\beta^{-n}}{M^n-1}\right]
\end{align}
with
\begin{align} \label{eq: def beta}
\beta = M^{\frac{1-\alpha}{2 \alpha}} > 1.
\end{align}
Then,
\begin{align}
\varepsilon_{X_n | Y_n} = \beta^{-n} \to 0,
\end{align}
and
\begin{align}
H_{\alpha}(X_n | Y_n) &= H_{\alpha}(X_n) \label{ITA1} \\[0.1cm]
&= \frac1{1-\alpha} \, \log \left( \bigl(1-\beta^{-n}\bigr)^{\alpha}
+ (M^n-1) \, \left(\frac{\beta^{-n}}{M^n-1}\right)^{\alpha} \right) \label{ITA2} \\[0.1cm]
&= \frac1{1-\alpha} \, \log \left( \bigl(1-\beta^{-n}\bigr)^{\alpha}
+ (M^n-1)^{1-\alpha} \, M^{-\frac{n(1-\alpha)}{2}} \right) \label{ITA3} \\[0.1cm]
&= \tfrac12 \, n \log M + \frac1{1-\alpha} \, \log \Bigl( (1-\beta^{-n})^{\alpha}
M^{-\frac{n(1-\alpha)}{2}} + (1-M^{-n})^{1-\alpha} \Bigr) \label{ITA4}
\end{align}
where \eqref{ITA1} holds since $Y_n$ is deterministic; \eqref{ITA2} follows from
\eqref{eq: Renyi entropy} and \eqref{eq: dist. X_n}; \eqref{ITA3} holds due to
\eqref{eq: def beta}. Consequently, since $\alpha \in (0,1)$, $M \geq 2$, and $\beta>1$,
the second term in the right side of \eqref{ITA4}
tends to~0. In conclusion, normalizing \eqref{ITA1}--\eqref{ITA4}
by $n$, and letting $n \to \infty$ yields
\begin{align}
\lim_{n \to \infty} \frac1n \, H_{\alpha}(X_n | Y_n) = \tfrac12 \log M.
\end{align}
\end{remark}

\begin{remark}
Theorem~\ref{thm:genhov}\ref{thm:genhov:a2}) is due to \cite[Theorem~15]{HoS-IT10}.
Furthermore, \cite[Example~2]{HoS-IT10} shows that Theorem~\ref{thm:genhov}\ref{thm:genhov:a2})
cannot be strengthened to $H(X_n | Y_n) \to 0$, in contrast to the case where
$\alpha \in (1, \infty]$ in Theorem~\ref{thm:genhov}\ref{thm:genhov:a}).
\end{remark}

\subsection{Explicit lower bounds on $\varepsilon_{X|Y}$} \label{subsec: explicit LBs}
The results in Section \ref{subsection: Generalized Fano inequality}
yield implicit lower bounds on $\varepsilon_{X|Y}$ as a function of the Arimoto-R\'enyi
conditional entropy. In this section, we obtain several explicit bounds.
As the following result shows, Theorem~\ref{theorem: generalized Fano inequality}
readily results in explicit lower bounds on $\varepsilon_{X|Y}$ as a function of
$H_{\frac12}(X|Y)$ and of $H_{2}(X|Y)$.

\vspace*{0.1cm}
\begin{theorem}
\label{proposition: closed-form lower bound on epsilon}
Let $X$ be a discrete random variable taking $M \geq 2$
possible values. Then,
\begin{align}
\label{eq: lower bound 1 on epsilon}
\varepsilon_{X|Y} & \geq \left(1-\frac1M\right) \, \frac1{\xi_1}
\left(1-\sqrt{\frac{\xi_1-1}{M-1}} \, \right)^2, \\[0.1cm]
\label{eq: lower bound 2 on epsilon}
\varepsilon_{X|Y} & \geq \left(1-\frac1M\right)
\left(1-\sqrt{\frac{\xi_2-1}{M-1}} \, \right)
\end{align}
where
\begin{align}
\label{eq: xi1}
& \xi_1 = M \exp\bigl(-H_{\frac12}(X|Y)\bigr), \\
\label{eq: xi2}
& \xi_2 = M \exp\bigl(-H_2(X|Y)\bigr).
\end{align}
\end{theorem}

\begin{proof}
Since $0 \leq H_2(X|Y) \leq H_{\frac12}(X|Y) \leq \log M$,
\eqref{eq: xi1}--\eqref{eq: xi2} imply
that $1 \leq \xi_1 \leq \xi_2 \leq M$.
To prove \eqref{eq: lower bound 1 on epsilon},
setting $\alpha=\tfrac12$ in \eqref{eq1: generalized Fano}
and using  \eqref{useful equality} with
$(\theta, s, t) = (M, M-1, \varepsilon_{X|Y})$ we obtain
\begin{align} \label{eq: from generalized Fano's ineq.}
\sqrt{1-\varepsilon_{X|Y}} + \sqrt{(M-1) \varepsilon_{X|Y}}
\geq \exp\left(\tfrac12 \, H_{\frac12}(X|Y)\right).
\end{align}
Substituting
\begin{align}
\label{eq: substitute1}
& v = \sqrt{\varepsilon_{X|Y}}, \\
\label{eq: substitute2}
& z = \exp\left(\tfrac12 \, H_{\frac12}(X|Y)\right)
\end{align}
yields the inequality
\begin{align} \label{eq: intermediate}
& \sqrt{1-v^2} \geq z - \sqrt{M-1} \, v.
\end{align}
If the right side of \eqref{eq: intermediate} is non-negative, then
\eqref{eq: intermediate} is transformed to the following quadratic
inequality in $v$:
\begin{align}
\label{eq: quadratic ineq.}
M v^2 - 2 \sqrt{M-1} \, zv + (z^2-1) \leq 0
\end{align}
which yields
\begin{align} \label{eq for v}
v \geq \frac1M \left( \sqrt{M-1} \, z - \sqrt{M-z^2} \right).
\end{align}
If, however, the right side of \eqref{eq: intermediate} is negative then
\begin{align} \label{eq: 2nd case}
v > \frac{z}{\sqrt{M-1}}
\end{align}
which implies the satisfiability of \eqref{eq for v} also in the latter case.
Hence, \eqref{eq for v} always holds. In view of \eqref{eq: xi1},
\eqref{eq: substitute1}, \eqref{eq: substitute2}, and since $\xi_1 \in [1,M]$,
it can be verified that the right side of \eqref{eq for v} is non-negative.
Squaring both sides of \eqref{eq for v}, and using \eqref{eq: xi1},
\eqref{eq: substitute1} and \eqref{eq: substitute2} give
\eqref{eq: lower bound 1 on epsilon}.

Similarly, setting $\alpha=2$ in \eqref{eq1: generalized Fano} and
using \eqref{useful equality} with
$(\theta, s, t) = (M, M-1, \varepsilon_{X|Y})$ yield a quadratic
inequality in $\varepsilon_{X|Y}$, from which
\eqref{eq: lower bound 2 on epsilon} follows.
\end{proof}

\begin{remark}
Following up on this work, Renes \cite{Renes-arXiv17} has generalized
\eqref{eq: lower bound 1 on epsilon} to the quantum setting.
\end{remark}

\begin{remark} \label{remark: Devijver74}
The corollary to Theorem~\ref{theorem: generalized Fano inequality}
in \eqref{eq: lower bound 2 on epsilon} is equivalent
to \cite[Theorem~3]{Devijver74}.
\end{remark}

\begin{remark} \label{remark: Kailath67}
Consider the special case where $X$ is an equiprobable binary
random variable, and $Y$ is a discrete random variable which
takes values on a set $\set{Y}$. Following the notation in
\cite[(7)]{Kailath67}, let $\rho \in [0, 1]$ denote the
Bhattacharyya coefficient
\begin{align}
\label{eq0: BC}
\rho & = \sum_{y \in \set{Y}} \sqrt{P_{Y|X}(y|0) P_{Y|X}(y|1)} \\
\label{eq1: BC}
& = 2 \sum_{y \in \set{Y}} P_Y(y) \sqrt{P_{X|Y}(0|y) \, P_{X|Y}(1|y)}
\end{align}
where \eqref{eq1: BC} holds due to Bayes' rule which implies that
$P_{Y|X}(y|x) = 2 P_{X|Y}(x|y) P_Y(y)$
for all $x \in \{0,1\}$ and $y \in \set{Y}$.
From \eqref{eq1: Arimoto - cond. RE} and \eqref{eq1: BC}, we obtain
\begin{align}
\label{eq1: Bhattacharyya}
& H_{\frac12}(X|Y) = \log(1+\rho).
\end{align}
Since $X$ is a binary random variable, it follows from
\eqref{eq: xi1} and \eqref{eq1: Bhattacharyya} that
$\xi_1 = \frac2{1+\rho}$; hence, the lower bound on the
minimal error probability in \eqref{eq: lower bound 1 on epsilon}
is given by
\begin{align} \label{eq2: Bhattacharyya}
\varepsilon_{X|Y} \geq \tfrac12 \left(1-\sqrt{1-\rho^2}\right)
\end{align}
recovering the bound in \cite[(49)]{Kailath67} (see also \cite{Toussiant72}).
\end{remark}

\begin{remark}
The lower bounds on $\varepsilon_{X|Y}$ in \eqref{eq: lower bound 1 on epsilon}
and \eqref{eq: lower bound 2 on epsilon} depend on $H_{\alpha}(X|Y)$ with $\alpha = \tfrac12$
and $\alpha=2$, respectively; due to their dependence on different orders $\alpha$, none of
these bounds is superseded by the other, and they both prove to be useful in view of
their relation to the conditional Bayesian distance in \cite[Definition~2]{Devijver74}
and the Bhattacharyya coefficient (see Remarks~\ref{remark: Devijver74} and~\ref{remark: Kailath67}).
\end{remark}

\begin{remark}
Taking the limit $M \to \infty$ in the right side of
\eqref{eq: lower bound 2 on epsilon} yields
\begin{align} \label{eq: Devijver - Th. 2}
\varepsilon_{X|Y} \geq 1 - \exp\left(-\tfrac12 \, H_2(X|Y)\right)
\end{align}
which is equivalent to \cite[(8)]{Vajda68} (see also
\cite[Theorem~2]{Devijver74}), and its loosening yields
\cite[Corollary~1]{Devijver74}.
Theorem~\ref{theorem: LB via RevHolderI} (see also
Theorem~\ref{theorem: LB via RevHolderI - list decoding})
tightens \eqref{eq: Devijver - Th. 2}.
\end{remark}

\begin{remark}
Arimoto \cite{Arimoto71} introduced a different generalization of entropy and
conditional entropy parameterized by a continuously differentiable function
$\mathtt{f} \colon (0,1] \to [0, \infty)$ satisfying $\mathtt{f} ( 1) = 0$:
\begin{align}\label{Arimoto-f-entropy}
H_{\mathtt{f}}(X)
&= \inf_Y \mathbb{E} \left[ \mathtt{f}( P_Y (X))\right], \\
H_{\mathtt{f}}(X|Y)
&= \mathbb{E} \left[ H_{\mathtt{f}}( P_{X|Y} (\cdot | Y)) \right],
\label{Arimoto-f-conditional-entropy}
\end{align}
where the infimum is over all the distributions defined on the same set as $X$.
Arimoto \cite[Theorem~3]{Arimoto71} went on to show the following generalization of
Fano's inequality:
\begin{align} \label{arimoto4.1}
H_{\mathtt{f}} (X|Y) \leq \min_{\theta \in (0,1)}
\left\{
\mathbb{P} [ X = Y ] \, \mathtt{f} ( 1 -\theta) +
\mathbb{P} [ X \neq Y ] \, \mathtt{f} \left(\frac{\theta}{M-1}\right)
\right\}.
\end{align}
A functional dependence can be established between the R\'enyi
entropy and $H_{\mathtt{f}}(X)$ for a certain choice of $\mathtt{f}$
(see \cite[Example~2]{Arimoto71}); in view of \eqref{eq2: Arimoto - cond. RE},
the Arimoto-R\'enyi conditional entropy can be expressed in terms of
$H_{\mathtt{f}}(X|Y) $, although the analysis of generalizing Fano's inequality
with the Arimoto-R\'enyi conditional entropy becomes rather convoluted following
this approach.
\end{remark}

For convenience, we assume throughout the rest of this subsection that
\begin{align} \label{tuesday}
P_{X|Y}(x|y) > 0, \quad
(x,y) \in \set{X} \times \set{Y}.
\end{align}
The following bound, which is a special case of
Theorem~\ref{theorem: LB via HolderI - list decoding}, involves
the Arimoto-R\'enyi conditional entropy of negative orders,
and bears some similarity to Arimoto's converse for channel coding
\cite{Arimoto73}.
\begin{theorem}
\label{theorem: LB via HolderI}
Let $P_{XY}$ be a probability measure defined on
$\set{X} \times \set{Y}$ with $|\set{X}|=M < \infty$, which satisfies \eqref{tuesday}.
For all $\alpha \in (-\infty, 0)$,
\begin{align} \label{eq: LB via HolderI}
\varepsilon_{X|Y} \geq \exp \left( \frac{1-\alpha}{\alpha} \,
\Bigl[ H_{\alpha}(X|Y) - \log(M-1) \Bigr] \right).
\end{align}
\end{theorem}
\par
It can be verified that the bound in \cite[(23)]{Tirza08} is equivalent to
\eqref{eq: LB via HolderI}. A different approach can be found in the proof of
Theorem~\ref{theorem: LB via HolderI - list decoding}.

\begin{remark}
By the assumption in \eqref{tuesday}, it follows from \eqref{P0N} that, for
$\alpha \in (-\infty, 0)$, the quantity in the exponent in the right side of
\eqref{eq: LB via HolderI} satisfies
\begin{align}
H_{\alpha}(X|Y) - \log(M-1)
& \geq H_0(X|Y) - \log(M-1) \\[0.1cm]
& = \log \frac{M}{M-1}.
\end{align}
Hence, by letting $\alpha \to 0$, the bound in \eqref{eq: LB via HolderI}
is trivial; while by letting $\alpha \to -\infty$, it follows from
\eqref{eq: cond. RE of order -infinity}, \eqref{tuesday} and
\eqref{eq: LB via HolderI} that
\begin{align}
\varepsilon_{X|Y} \geq (M-1) \, \expectation
\Bigl[ \min_{x \in \set{X}} P_{X|Y}(x|Y) \Bigr].
\end{align}
The  $\alpha \in (-\infty, 0)$ that results in the tightest bound
in \eqref{eq: LB via HolderI} is examined numerically in Example~\ref{example: negative alpha},
which illustrates the utility of Arimoto-R\'enyi conditional entropies of negative orders.
\end{remark}
\begin{remark}
For binary hypothesis testing, the lower bound on $\varepsilon_{X|Y}$ in
\eqref{eq: LB via HolderI} is asymptotically tight by letting $\alpha \to -\infty$
since, in view of \eqref{inf-inf} and \eqref{varepsilon_X|Y},
\begin{align} \label{eq: Galil3}
\varepsilon_{X|Y} = \exp \bigl( -H_{-\infty}(X|Y) \bigr).
\end{align}
\end{remark}

Next, we provide a lower bound depending on the Arimoto-R\'enyi conditional entropy of orders
greater than~1. A more general version of this result is given in
Theorem~\ref{theorem: LB via RevHolderI - list decoding} by relying on our generalization of
Fano's inequality for list decoding.
\begin{theorem}
\label{theorem: LB via RevHolderI}
Let $P_{XY}$ be a probability measure defined on
$\set{X} \times \set{Y}$ which satisfies \eqref{tuesday}, with $\set{X}$ being finite
or countably infinite. For all $\alpha \in (1, \infty)$
\begin{align} \label{eq: LB via RevHolderI}
\varepsilon_{X|Y} \geq 1 - \exp \left(\frac{1-\alpha}{\alpha} \; H_{\alpha}(X|Y) \right).
\end{align}
\end{theorem}

\begin{proof}
In view of the monotonicity property for positive orders in \eqref{mon}, we obtain that
\begin{align}  \label{eq: LB on H infinity}
H_{\infty}(X|Y) &\geq \frac{\alpha-1}{\alpha} \; H_{\alpha}(X|Y).
\end{align}
and the desired result follows in view of \eqref{varepsilon_X|Y}. (Note that for
$\alpha \in (0,1]$, the right side of \eqref{eq: LB via RevHolderI} is nonpositive.)
\end{proof}

\begin{remark}
The implicit lower bound on $\varepsilon_{X|Y}$ given in \eqref{th3-5} (same as
\eqref{eq1: generalized Fano}) is tighter than the explicit lower bound in \eqref{eq: LB via RevHolderI}.
\end{remark}

\begin{remark}
The lower bounds on $\varepsilon_{X|Y}$ in Theorems~\ref{theorem: generalized Fano inequality}
and~\ref{theorem: LB via RevHolderI}, which hold for positive orders of $H_{\alpha}(X|Y)$, are both
asymptotically tight by letting $\alpha \to \infty$. In contrast, the lower bound on $\varepsilon_{X|Y}$
of Theorem~\ref{theorem: LB via HolderI}, which holds for negative orders of $H_{\alpha}(X|Y)$, is
not asymptotically tight by letting $\alpha \to -\infty$ (unless $X$ is a binary random variable).
\end{remark}

In the following example, the lower bounds on $\varepsilon_{X|Y}$ in
Theorems~\ref{theorem: generalized Fano inequality} and~\ref{theorem: LB via RevHolderI}
are examined numerically in their common range of $\alpha \in (1, \infty)$.
\begin{example}\label{example:45b}
Let $X$ and $Y$ be random variables defined on the set $\set{X} = \{1, 2, 3\}$,
and let
\begin{align} \label{P}
\bigl[P_{XY}(x,y)\bigr]_{(x,y) \in \set{X}^2}
= \frac1{45} \left( \begin{array}{ccc}
8  & 1  & 6 \\
3  & 5  & 7 \\
4  & 9  & 2
\end{array}
\right).
\end{align}
It can be verified that
$\varepsilon_{X|Y} = \tfrac{21}{45} \approx 0.4667$.
Note that although in this example the bound in \eqref{eq: LB via RevHolderI}
is only slightly looser than the bound in \eqref{eq1: generalized Fano} for
moderate values of $\alpha > 1$ (see Table~\ref{table1})\footnote{Recall
that for $\alpha = 2$, \eqref{eq1: generalized Fano} admits the
explicit expression in \eqref{eq: lower bound 2 on epsilon}.}, both
are indeed asymptotically tight as $\alpha \to \infty$; furthermore,
\eqref{eq: LB via RevHolderI} has the advantage of providing a
closed-form lower bound on $\varepsilon_{X|Y}$ as a function of
$H_{\alpha}(X|Y)$ for $\alpha \in (1, \infty)$.

\begin{table}[h]
\renewcommand{\arraystretch}{1.5}
\begin{center}
\begin{tabular}{||r|c|c||}
\hline
$\alpha$ & \eqref{eq1: generalized Fano} &\eqref{eq: LB via RevHolderI}  \\
\hline
   2 & 0.4247 & 0.3508  \\
   4 & 0.4480 & 0.4406  \\
   6 & 0.4573 & 0.4562  \\
   8 & 0.4620 & 0.4613  \\
  10 & 0.4640 & 0.4635  \\
  50 & 0.4667 & 0.4667  \\
\hline
\end{tabular}
\vspace*{0.5cm}
\caption{\label{table1} Lower bounds on $\varepsilon_{X|Y}$  in
\eqref{eq1: generalized Fano} and \eqref{eq: LB via RevHolderI}
for Example~\ref{example:45b}.}
\end{center}
\end{table}
\end{example}

\begin{example} \label{example: negative alpha}
Let $X$ and $Y$ be random variables defined on the set $\set{X} = \{1, 2, 3, 4\}$,
and let
\begin{align} \label{P6X6}
\bigl[P_{XY}(x,y)\bigr]_{(x,y) \in \set{X}^2}
= \frac{1}{400}\left( \begin{array}{cccc}
10  &  38  &  10  &  26 \\
32  &  20  &  44  &  20 \\
10  &  29  &  10  &  35 \\
41  &  20  &  35  &  20
\end{array}
\right).
\end{align}
In this case $\varepsilon_{X|Y} = \tfrac{121}{200} = 0.6050$, and the tightest
lower bound in \eqref{eq: LB via HolderI} for $\alpha \in (-\infty, 0)$ is equal to
0.4877 (obtained at $ \alpha = - 2.531$).
Although $\varepsilon_{X|Y}$ can be calculated exactly when $H_{\infty}(X|Y)$ is
known (see \eqref{varepsilon_X|Y}), this example illustrates that
conditional Arimoto-R\'{e}nyi entropies of negative orders are useful in the
sense that \eqref{eq: LB via HolderI} gives an informative lower bound on
$\varepsilon_{X|Y}$ by knowing $H_{\alpha}(X|Y)$ for a negative $\alpha$.
\end{example}

\subsection{List decoding}
In this section we consider the case where the decision rule outputs a list of
choices. The extension of Fano's inequality to list decoding was initiated in
\cite[Section~5]{Ahlswede_Gacs_Korner} (see also \cite[Appendix~3.E]{FnT14}).
It is useful for proving converse results in conjunction with the blowing-up
lemma \cite[Lemma~1.5.4]{Csiszar_Korner}.
The main idea of the successful combination of these two tools is that, given an
arbitrary code, one can blow-up the decoding sets in such a way that the probability
of decoding error can be as small as desired for sufficiently large blocklengths;
since the blown-up decoding sets are no longer disjoint, the resulting setup is
a list decoder with subexponential list size.
\par
A generalization  of Fano's inequality for list decoding of size $L$  is
\cite{verdubook}\footnote{See \cite[Lemma~1]{Kim_Sutivong_Cover}
for a weaker version of \eqref{fanoML}.}
\begin{align}\label{fanoML}
H (X | Y ) \leq \log M - d\bigl( P_{\set{L}} \| 1-\tfrac{L}{M} \bigr),
\end{align}
where $P_{\mathcal{L}}$ denotes the probability of $X$ not being in the list.
As we noted before,
averaging a conditional version of \eqref{eq: generalized UB on RE}
with respect to the observation is not viable in the case of $H_\alpha(X|Y)$
with $\alpha \neq 1$ (see \eqref{eq1: Arimoto - cond. RE}). A pleasing
generalization of \eqref{fanoML} to the Arimoto-R\'enyi conditional entropy
does indeed hold as the following result shows.

\begin{theorem}
\label{theorem: generalized Fano-Renyi - list decoding}
Let $P_{XY}$ be a probability measure defined on
$\set{X} \times \set{Y}$ with $|\set{X}|=M$. Consider
a decision rule\footnote{$\binom{\set{X}}{L}$ stands for the set of  all the subsets
of $\set{X}$ with cardinality $L$, with $L \leq | \set{X}|$.}
$\mathcal{L} \colon \set{Y} \to \binom{\set{X}}{L}$, and denote
the decoding error probability by
\begin{align}
P_{\mathcal{L}} = \prob \bigl[ X \notin \mathcal{L}(Y) \bigr].
\end{align}
Then, for all $\alpha \in (0,1) \cup(1,\infty)$,
\begin{align} \label{eq: generalized Fano-Renyi - list decoding}
H_{\alpha}(X | Y) & \leq \log M - d_{\alpha}\bigl( P_{\mathcal{L}}
\| 1-\tfrac{L}{M} \bigr) \\
\label{eq2: generalized Fano-Renyi - list decoding}
& = \frac1{1-\alpha} \;
\log \Bigl( L^{1-\alpha} \, \bigl(1-P_{{\mathcal{L}}}\bigr)^{\alpha}
+ (M-L)^{1-\alpha} \, P_{{\mathcal{L}}}^{\alpha} \Bigr)
\end{align}
with equality in \eqref{eq: generalized Fano-Renyi - list decoding}
if and only if
\begin{align} \label{eq: tight Fano-Renyi - list decoding}
P_{X|Y}(x|y) =
\begin{dcases}
\frac{P_{\mathcal{L}}}{M-L}, & \quad x \notin \set{L}(y) \\[0.2cm]
\frac{1-P_{\mathcal{L}}}{L}, & \quad x \in \set{L}(y).
\end{dcases}
\end{align}
\end{theorem}

\begin{proof}
Instead of giving a standalone proof, for brevity, we explain
the differences between this and the proof of
Theorem~\ref{theorem: generalized Fano inequality}:
\begin{itemize}
\item
Instead of the conditional version of \eqref{eq: specialized UB on RE},
we use a conditional version of \eqref{eq: gen. UB on RE} given the observation $Y=y$;
\item
The choices of the arguments of the function $f_{\alpha, \beta, \gamma}$ in
\eqref{th3-2}--\eqref{th3-4} namely,
$\beta = (M-1)^{1-\alpha}$ and $\gamma=1$,
are replaced by $\beta = (M-L)^{1-\alpha}$ and $\gamma = L^{1-\alpha}$;
\item \eqref{eq2: generalized Fano-Renyi - list decoding}
follows from \eqref{useful equality} with $(\theta, s, t) = (M, M-L, P_{\set{L}})$;
\item
Equality in \eqref{eq: generalized Fano-Renyi - list decoding} holds if and only if
the condition for equality in the statement of Corollary~\ref{corollary: generalized UB of RE},
conditioned on the observation, is satisfied; the latter implies that, given $Y=y$, $X$ is
equiprobable on both $\set{L}(y)$ and $\set{L}^{\mathrm{c}}(y)$ for all $y \in \set{Y}$.
\end{itemize}
\end{proof}

Next, we give the fixed list-size generalization of Theorem~\ref{theorem: LB via HolderI}.
\begin{theorem}
\label{theorem: LB via HolderI - list decoding}
Let $P_{XY}$ be a probability measure defined on
$\set{X} \times \set{Y}$ with $|\set{X}|=M < \infty$, which satisfies \eqref{tuesday},
and let $\mathcal{L} \colon \set{Y} \to \binom{\set{X}}{L}$.
Then, for all $\alpha \in (-\infty, 0)$, the probability that the decoding list does not include the
correct decision satisfies
\begin{align} \label{eq: LB via HolderI - list decoding}
P_{\mathcal{L}} \geq \exp \left( \frac{1-\alpha}{\alpha} \,
\Bigl[ H_{\alpha}(X|Y) - \log(M-L) \Bigr] \right).
\end{align}
\end{theorem}
\begin{proof}
Let $\xi \colon \set{X} \times \set{Y} \to \{0,1\}$ be given by the indicator function
\begin{align} \label{eq: xi}
\xi(x,y) = 1\bigl\{x \notin \mathcal{L}(y)\bigr\},
\end{align}
Let $u \colon \set{X} \times \set{Y} \to [0, \infty)$, and $\beta>1$.
By H\"older's inequality,
\begin{align}   \label{ISR2}
& \mathbb{E} \bigl[ \xi(X,Y) \, u(X,Y) \bigr]
\leq \mathbb{E}^{\frac1\beta} \bigl[ \xi^\beta(X,Y) \bigr] \;
\mathbb{E}^{\frac{\beta-1}{\beta}} \bigl[ u^{\frac{\beta}{\beta-1}}(X,Y) \bigr].
\end{align}
Therefore,
\begin{align}
P_{\set{L}} &=
\mathbb{E} \bigl[ \xi^{\beta}(X,Y) \bigr] \\
&\geq \mathbb{E}^{\beta} \bigl[ \xi(X,Y) \, u(X,Y) \bigr] \;
\mathbb{E}^{1-\beta}  \bigl[ u^{\frac{\beta}{\beta-1}}(X,Y) \bigr]. \label{ISR3}
\end{align}
Under the assumption in \eqref{tuesday}, we specialize \eqref{ISR3} to
\begin{align}
u(x,y) = \frac1{P_{X|Y}(x|y)} \, \left( \sum_{x' \in \set{X}}
P_{X|Y}^{\frac1{1-\beta}}(x'|y) \right)^{1-\beta}
\label{ISR4}
\end{align}
for all $(x,y) \in \set{X} \times \set{Y}$.
From \eqref{eq1: Arimoto - cond. RE}, \eqref{eq: xi} and \eqref{ISR4}, we have
\begin{align}
\mathbb{E}\bigl[ \xi(X,Y) \, u(X,Y) \bigr]
& = \mathbb{E} \Bigl[ \mathbb{E}\bigl[ \xi(X,Y) \, u(X,Y) \, \big| \, Y \bigr] \Bigr] \\
& = \mathbb{E} \left[ \sum_{x \in \set{X}} P_{X|Y}(x|Y) \, u(x,Y) \, \xi(x,Y) \right] \\
& = \mathbb{E} \left[ \left( \sum_{x' \in \set{X}} P_{X|Y}^{\frac1{1-\beta}}(x'|Y) \right)^{1-\beta}
\sum_{x \in \set{X}} \xi(x,Y) \right] \\
& = \mathbb{E} \left[ \left( \sum_{x' \in \set{X}} P_{X|Y}^{\frac1{1-\beta}}(x'|Y)
\right)^{1-\beta} \, |\mathcal{L}^{\mathrm{c}}(Y)|   \right] \\
&= (M-L) \exp\left( -\beta \, H_{\frac1{1-\beta}}(X|Y) \right)  \label{ISR5}
\end{align}
and
\begin{align}
\mathbb{E}\bigl[ u^{\frac{\beta}{\beta-1}}(X,Y) \bigr]
&= \mathbb{E} \Bigl[ \mathbb{E}\bigl[ u^{\frac{\beta}{\beta-1}}(X,Y) \, \big| \, Y \bigr] \Bigr] \\
&= \mathbb{E} \left[ \sum_{x \in \set{X}}  P_{X|Y}(x|Y) \, u^{\frac{\beta}{\beta-1}}(x,Y) \right] \\
&= \mathbb{E} \left[ \left( \sum_{x \in \set{X}} P_{X|Y}^{\frac1{1-\beta}}(x|Y) \right)^{1-\beta} \right] \\
&= \exp\bigl( -\beta \, H_{\frac1{1-\beta}}(X|Y) \bigr).   \label{ISR6}
\end{align}
Assembling \eqref{ISR3}, \eqref{ISR5} and \eqref{ISR6} and substituting
$\alpha = \frac1{1-\beta}$ results  in \eqref{eq: LB via HolderI - list decoding}.
\end{proof}

The fixed list-size generalization of Theorem~\ref{theorem: LB via RevHolderI}
is the following.
\begin{theorem}
\label{theorem: LB via RevHolderI - list decoding}
Let $P_{XY}$ be a probability measure defined on
$\set{X} \times \set{Y}$ which satisfies \eqref{tuesday}, with $\set{X}$ being finite
or countably infinite, and let $\mathcal{L} \colon \set{Y} \to \binom{\set{X}}{L}$.
Then, for all $\alpha \in (1, \infty)$,
\begin{align} \label{eq: LB via RevHolderI - list decoding}
P_{\mathcal{L}} \geq 1 - \exp \left(\frac{1-\alpha}{\alpha}
\, \Bigl[ H_{\alpha}(X|Y) - \log L \Bigr] \right).
\end{align}
\end{theorem}
\begin{proof}
From \eqref{eq2: generalized Fano-Renyi - list decoding}, for all $\alpha \in (1, \infty)$,
\begin{align}
H_{\alpha}(X|Y) &\leq \frac1{1-\alpha} \, \log\bigl(L^{1-\alpha} \, (1-P_{\mathcal{L}})^\alpha \bigr) \\
&= \log L + \frac{\alpha}{1-\alpha} \, \log(1-P_{\mathcal{L}})
\end{align}
which gives the bound in \eqref{eq: LB via RevHolderI - list decoding}.
\end{proof}

\begin{remark}
The implicit lower bound on $\varepsilon_{X|Y}$ given by the generalized Fano's inequality
in \eqref{eq: generalized Fano-Renyi - list decoding} is tighter than the explicit
lower bound in \eqref{eq: LB via RevHolderI - list decoding} (this can be verified from the
proof of Theorem~\ref{theorem: LB via RevHolderI - list decoding}).
\end{remark}

\subsection{Lower bounds on the Arimoto-R\'{e}nyi conditional entropy}
\label{subsection: LB on cond. RE}
The major existing lower bounds on the Shannon conditional entropy $H(X|Y)$
as a function of the minimum error probability $\varepsilon_{X|Y}$ are:
\begin{enumerate}
\item In view of \cite[Theorem~11]{HoS-IT10}, \eqref{onehalf} (shown in
\cite[Theorem~1]{Baladova66}, \cite[(12)]{ChuC66} and \cite[(41)]{HellmanR70}
for finite alphabets) holds for a general discrete random variable $X$.
As an example where \eqref{onehalf} holds with equality,
consider a Z-channel with input and output alphabets $\set{X}=\set{Y}=\{0,1\}$,
and assume that $P_X(0)=1-P_X(1)=p \in \bigl(0, \tfrac12\bigr]$, and
\begin{align}
& P_{Y|X}(0|0)=1, \quad P_{Y|X}(0|1)=\frac{p}{1-p}.
\end{align}
It can be verified that $\varepsilon_{X|Y} = p$, and $H(X|Y) = 2p$ bits.
\item Due to Kovalevsky \cite{Kovalevsky}, Tebbe and Dwyer \cite{TebbeD68}
(see also \cite{FederM94}) in the finite alphabet case, and to Ho and Verd\'u
\cite[(109)]{HoS-IT10} in the general case,
\begin{align} \label{eq: LB cond. Shannon entropy}
\phi\bigl( \varepsilon_{X|Y} \bigr) \leq H(X|Y)
\end{align}
where $\phi \colon [0,1) \to [0, \infty)$ is the piecewise linear function
that is defined on the interval $t \in \bigl[1-\frac1k, 1-\frac1{k+1} \bigr)$ as
\begin{align} \label{eq: phi}
\phi(t) = t \, k (k+1) \log \left( \frac{k+1}{k} \right) + (1-k^2) \log(k+1) + k^2 \log k
\end{align}
where $k $ is an arbitrary positive integer.
Note that \eqref{eq: LB cond. Shannon entropy} is tighter than \eqref{onehalf} since $\phi(t) \geq 2 t \log 2 $.
\end{enumerate}

In view of \eqref{varepsilon_X|Y}, since $H_{\alpha}(X|Y)$ is monotonically decreasing
in $\alpha$, one can readily obtain the following counterpart to Theorem~\ref{theorem: LB via RevHolderI}:
\begin{align} \label{varepsilon_X|Y+}
H_{\alpha}(X | Y) \geq \log \frac{1}{1- \varepsilon_{X|Y}}
\end{align}
for $\alpha \in [0, \infty]$ with equality if $\alpha = \infty$.

\begin{remark}
As a consequence of \eqref{varepsilon_X|Y+},
it follows that if $\alpha \in [0,\infty)$ and $\{X_n\}$ is a sequence of discrete
random variables, then $H_{\alpha}(X_n | Y_n) \to 0$
implies that $\varepsilon_{X_n | Y_n} \to 0$, thereby generalizing the case $\alpha=1$
shown in \cite[Theorem~14]{HoS-IT10}.
\end{remark}

\begin{remark}
Setting $\alpha=2$ in \eqref{varepsilon_X|Y+} yields
\cite[Theorem~1]{Devijver74}, and setting $\alpha=1$ in
\eqref{varepsilon_X|Y+} yields \eqref{mystery}.
Recently, Prasad \cite[Section~2.C]{Prasad15} improved the latter bound
in terms of the entropy of tilted
probability measures, instead of $H_{\alpha}(X|Y)$.
\end{remark}

\par
The next result gives a counterpart to Theorem~\ref{theorem: generalized Fano inequality},
and a generalization of \eqref{eq: LB cond. Shannon entropy}.
\begin{theorem} \label{theorem: LB on the conditional Renyi entropy}
If $\alpha \in (0,1) \cup (1, \infty)$, then
\begin{align} \label{eq: LB on the cond. RE}
\frac{\alpha}{1-\alpha} \, \log  g_{\alpha}( \varepsilon_{X|Y})  \leq H_{\alpha}(X|Y) ,
\end{align}
where the piecewise linear function
$g_{\alpha} \colon [0,1) \to \set{D}_{\alpha}$, with $\set{D}_{\alpha} = [1, \infty)$
for $\alpha \in (0,1)$ and $\set{D}_{\alpha} = (0,1]$ for $\alpha \in (1, \infty)$, is defined by
\begin{align}
\label{eq: g_alpha}
g_{\alpha}(t) =  \left( k(k+1)^{\frac1\alpha}-k^{\frac1\alpha}(k+1) \right) t
+  k^{\frac1{\alpha}+1} - (k-1) (k+1)^{\frac1{\alpha}}
\end{align}
on the interval $t \in \bigl[1-\frac1k, 1-\frac1{k+1} \bigr)$ for an arbitrary positive
integer $k$.
\end{theorem}

\begin{proof}
For every $y \in \set{Y}$ such that $P_Y(y) > 0$, Theorem~\ref{theorem: LB of RE}
and \eqref{Aachen} yield
\begin{align}  \label{eq: LB on cond. RE given Y=y}
H_{\alpha}(X \, | \, Y=y)
\geq s_{\alpha} \bigl( \varepsilon_{X|Y}(y) \bigr)
\end{align}
where $s_{\alpha} \colon [0,1) \to [0, \infty)$ is given by
\begin{align}
s_{\alpha}(t) = \frac1{1-\alpha} \, \log \left( \left\lfloor \frac1{1-t}
\right\rfloor (1-t)^{\alpha} + \left(1-(1-t) \left\lfloor \frac1{1-t}
\right\rfloor \right)^{\alpha}\right).
\end{align}
In the remainder of the proof, we consider the cases $\alpha \in (0,1)$ and $\alpha \in (1, \infty)$
separately.
\begin{itemize}
\item
$\alpha \in (0,1)$.
Define the function $f_{\alpha} \colon [0,1) \to [1, \infty)$ as
\begin{align}\label{vier}
f_{\alpha}(t) & = \exp \left( \frac{1-\alpha}{\alpha} \; s_{\alpha}(t) \right).
\end{align}
The piecewise linear function $g_{\alpha}(\cdot)$ in \eqref{eq: g_alpha}
coincides with the monotonically increasing function $f_{\alpha}(\cdot)$ in \eqref{vier} at $t_k = 1-\frac1k$
for every positive integer $k$ (note that $f_{\alpha}(t_k) = g_{\alpha}(t_k) = k^{\frac1\alpha-1}$).
It can be also verified that $g_{\alpha}(\cdot)$ is the lower convex envelope of $f_{\alpha}(\cdot)$
(i.e., $g_{\alpha}(\cdot)$ is the largest convex function for which $g_\alpha(t) \leq f_\alpha(t)$
for all $t \in [0,1)$). Hence,
\begin{align}
\label{eq: ee1}
H_{\alpha}(X|Y) &\geq \frac{\alpha}{1-\alpha} \, \log \mathbb{E}
\left[ f_{\alpha} \bigl( \varepsilon_{X|Y}(Y) \bigr) \right] \\[0.1cm]
&\geq \frac{\alpha}{1-\alpha} \, \log \mathbb{E}
\left[ g_{\alpha} \bigl( \varepsilon_{X|Y}(Y) \bigr) \right]
\label{eq: ee4}\\[0.1cm]
& \geq \frac{\alpha}{1-\alpha} \, \log \, g_{\alpha}( \varepsilon_{X|Y})
\label{eq: ee5}
\end{align}
where \eqref{eq: ee1} follows from \eqref{eq2: Arimoto - cond. RE},
\eqref{eq: LB on cond. RE given Y=y}, \eqref{vier}, and since
$f_{\alpha} \colon [0,1) \to [1, \infty)$ is a monotonically
increasing function for $\alpha \in (0,1)$;
\eqref{eq: ee4} follows from $f_\alpha \geq g_\alpha$; \eqref{eq1: avg. cond. epsilon}
and Jensen's inequality for the convex function $g_{\alpha}(\cdot)$ result in \eqref{eq: ee5}.
\item
$\alpha \in (1, \infty)$. The function $f_{\alpha}(\cdot)$ in \eqref{vier} is monotonically decreasing
in $[0,1)$, and the piecewise linear function $g_{\alpha}(\cdot)$ in \eqref{eq: g_alpha} coincides with
$f_{\alpha}(\cdot)$ at $t_k = 1-\frac1k$ for every positive integer $k$. Furthermore, $g_{\alpha}(\cdot)$ is
the smallest concave function that is not below $f_{\alpha}(\cdot)$ on the interval $[0,1)$. The proof
now proceeds as in the previous case, and \eqref{eq: ee1}--\eqref{eq: ee5} continue to hold for $\alpha>1$.
\end{itemize}
\end{proof}

\begin{remark}
The most useful domain of applicability of Theorem~\ref{theorem: LB on the conditional Renyi entropy}
is $\varepsilon_{X|Y} \in [0,\tfrac12]$, in which case the lower bound specializes to ($k=1$)
\begin{align} \label{eq: Barcelona}
\frac{\alpha}{1 - \alpha} \,
\log \Bigl( 1 + \bigl( 2^\frac1{\alpha} -2 \bigr) \varepsilon_{X|Y} \Bigr) \leq H_\alpha (X|Y)
\end{align}
which yields \eqref{onehalf} as $\alpha \to 1$.
\end{remark}

\begin{remark}
Theorem~\ref{theorem: LB on the conditional Renyi entropy} is indeed a generalization of
\eqref{eq: LB cond. Shannon entropy} since for all $\tau \in[0,1]$,
\begin{align}
\lim_{\alpha \to 1}  \frac{\alpha}{1-\alpha} \, \log \, g_{\alpha}( \tau) =  \phi (\tau) ,
\end{align}
with $\phi$ defined in \eqref{eq: phi}.
\end{remark}

\begin{remark}
As $\alpha \to \infty$, \eqref{eq: LB on the cond. RE} is asymptotically tight:
from \eqref{eq: g_alpha},
\begin{align}
\lim_{\alpha \to \infty} g_{\alpha}(t)
&= 1-t \label{eq: lim g_alpha}
\end{align}
so the right side of
\eqref{eq: LB on the cond. RE} converges to $\log \frac1{1-\varepsilon_{X|Y}}$
which, recalling \eqref{varepsilon_X|Y}, proves the claim.
\end{remark}

\begin{remark}
It can be shown that the function in \eqref{eq: g_alpha} satisfies
\begin{align}
g_{\alpha}(t)
\begin{dcases}
\leq (1-t)^{1-\frac1\alpha}, & \alpha \in (1, \infty] \\[0.2cm]
\geq (1-t)^{1-\frac1\alpha}, & \alpha \in (0, 1).
\end{dcases}
\end{align}
Moreover, if $t = 1 - \frac1k$, for $k\in \{ 1, 2, 3, \ldots \}$, then
\begin{align}
g_{\alpha}(t) = (1-t)^{1- \frac1\alpha}.
\end{align}
Therefore, Theorem~\ref{theorem: LB on the conditional Renyi entropy}
gives a tighter bound than \eqref{varepsilon_X|Y+}, unless
$\varepsilon_{X|Y} \in \bigl\{ \frac12, \frac23, \frac34 , \ldots \frac{M-1}{M} \bigr\}$
($M$ is allowed to be $\infty$ here) in which case they are identical,
and independent of $\alpha$ as it is illustrated in Figure~\ref{figure: plot_UB/LB_Cond_RE}.
\end{remark}

\vspace*{-4cm}
\begin{figure}[h]
\centerline{\includegraphics[width=10cm]{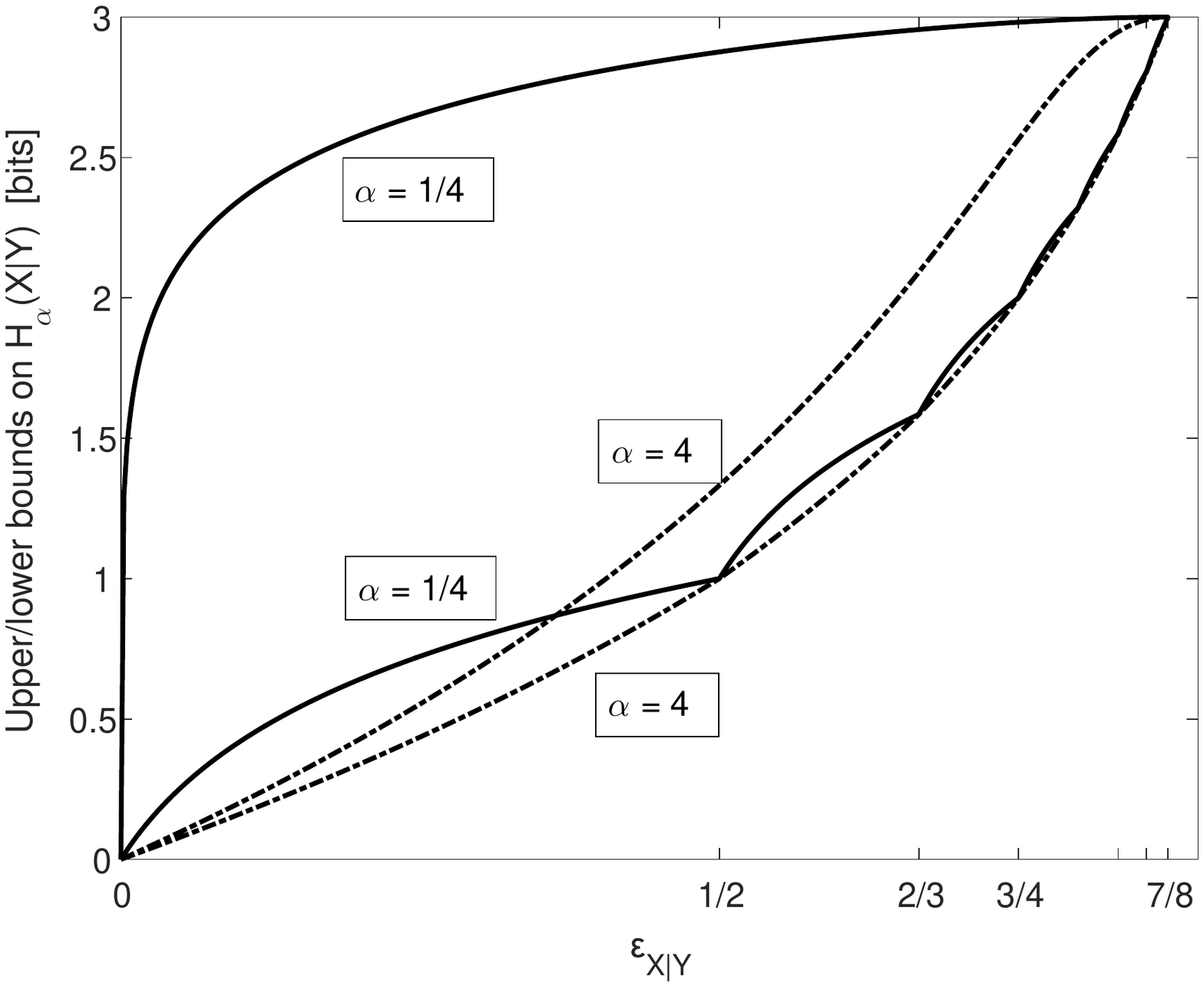}}
\vspace*{-3.2cm}
\caption{\label{figure: plot_UB/LB_Cond_RE}
Upper and lower bounds on $H_{\alpha}(X|Y)$ in Theorems~\ref{theorem: generalized Fano inequality}
and~\ref{theorem: LB on the conditional Renyi entropy}, respectively,
as a function of $\varepsilon_{X|Y} \in [0, 1-\tfrac1M]$ for $\alpha = \tfrac14$ (solid lines)
and $\alpha=4$ (dash-dotted lines) with $M=8$.}
\end{figure}

The following proposition applies when $M$ is finite, but does not depend on $M$. It is
supported by Figure~\ref{figure: plot_UB/LB_Cond_RE}.
\begin{proposition} \label{prop: lb/ub}
Let $M \in \{2, 3, \ldots \}$ be finite, and let the upper and lower bounds on $H_{\alpha}(X|Y)$
as a function of $\varepsilon_{X|Y}$, as given in Theorems~\ref{theorem: generalized Fano inequality}
and~\ref{theorem: LB on the conditional Renyi entropy}, be denoted by $u_{\alpha,M}(\cdot)$ and
$l_{\alpha}(\cdot)$, respectively. Then,
\begin{enumerate}[a)]
\item \label{prop: lb/ub - part 1}
these bounds coincide if and only if $X$ is a deterministic function of the observation
$Y$ or $X$ is equiprobable on the set $\set{X}$ and independent of $Y$;
\item \label{prop: lb/ub - part 2}
the limit of the ratio of the upper-to-lower bounds
when $\varepsilon_{X|Y} \to 0$ is given by
\begin{align} \label{eq: limit u/l}
\lim_{\varepsilon_{X|Y} \to 0}  \frac{u_{\alpha,M}(\varepsilon_{X|Y})}{l_{\alpha}(\varepsilon_{X|Y})} =
\begin{dcases}
\hspace*{0.4cm} \infty, & \quad \alpha \in (0,1), \\[0.2cm]
\frac1{2 - 2^{\frac1\alpha}}, & \quad \alpha \in (1, \infty).
\end{dcases}
\end{align}
\end{enumerate}
\end{proposition}
\begin{proof}
See Appendix~\ref{app: lb/ub}.
\end{proof}

\begin{remark}
The low error probability limit of the ratio of upper-to-lower bounds
in \eqref{eq: limit u/l} decreases monotonically with $\alpha \in (1,\infty)
$ from $\infty$ to $1$.
\end{remark}

\subsection{Explicit upper bounds on $\varepsilon_{X|Y}$}

In this section, we give counterparts to the explicit lower bounds on
$\varepsilon_{X|Y}$ given in Section~\ref{subsec: explicit LBs} by
capitalizing on the bounds in Section~\ref{subsection: LB on cond. RE}.

The following result is a consequence of Theorem~\ref{theorem: LB on the conditional Renyi entropy}:
\begin{theorem} \label{theorem: UB min err}
Let $k \in \naturals$, and $\alpha \in (0,1) \cup (1, \infty)$.
If $\log k \leq H_{\alpha}(X|Y) < \log(k+1)$, then
\begin{align} \label{eq: UB1 min err}
\varepsilon_{X|Y} \leq
\frac{\exp\left(\frac{1-\alpha}{\alpha} \, H_{\alpha}(X|Y) \right)
- k^{\frac1{\alpha}+1} + (k-1) (k+1)^{\frac1{\alpha}} }{
k(k+1)^{\frac1\alpha}-k^{\frac1\alpha} (k+1)}.
\end{align}
Furthermore, the upper bound on $\varepsilon_{X|Y}$ as a function of $H_{\alpha}(X|Y)$
is asymptotically tight in the limit where $\alpha \to \infty$.
\end{theorem}

\begin{proof}
The left side of \eqref{eq: LB on the cond. RE} is equal to $\log k$
when $\varepsilon_{X|Y} = 1-\frac1{k}$ for $k \in \naturals$, and it is also
monotonically increasing in $\varepsilon_{X|Y}$. Hence, if
$\varepsilon_{X|Y} \in \bigl[1-\frac1k, 1-\frac1{k+1}\bigr)$, then the lower
bound on $H_{\alpha}(X|Y)$ in the left side of \eqref{eq: LB on the cond. RE}
lies in the interval $[\log k, \log(k+1))$. The bound in \eqref{eq: UB1 min err}
can be therefore verified to be equivalent to
Theorem~\ref{theorem: LB on the conditional Renyi entropy}. Finally, the
asymptotic tightness of \eqref{eq: UB1 min err} in the limit where
$\alpha \to \infty$ is inherited by the same asymptotic tightness of
the equivalent bound in Theorem~\ref{theorem: LB on the conditional Renyi entropy}.
\end{proof}

\begin{remark}  \label{remark: UB min err - alpha=1}
By letting $\alpha \to 1$ in the right side of \eqref{eq: UB1 min err},
we recover the bound by Ho and Verd\'u \cite[(109)]{HoS-IT10}:
\begin{align} \label{eq: UB2 min err}
\varepsilon_{X|Y} \leq
\frac{H(X|Y) + (k^2-1) \log(k+1) - k^2 \log k}{k(k+1) \log \left(\frac{k+1}{k}\right)}
\end{align}
if $\log k \leq H(X|Y) < \log(k+1)$ for an arbitrary $k \in \naturals$. In the special
case of finite alphabets, this result was obtained in \cite{FederM94}, \cite{Kovalevsky}
and \cite{TebbeD68}. In the finite alphabet setting, Prasad \cite[Section~5]{Prasad15}
recently refined the bound in \eqref{eq: UB2 min err} by lower bounding $H(X|Y)$ subject
to the knowledge of the first two largest posterior probabilities rather than only the
largest one; following the same approach, \cite[Section~6]{Prasad15} gives a refinement of
Fano's inequality.
\end{remark}

\indent
\textit{Example} \ref{example:45b} \textit{(cont.)}
Table~\ref{table2} compares the asymptotically tight bounds in
Theorems~\ref{theorem: generalized Fano inequality} and
\ref{theorem: UB min err} for three values of $\alpha$.
\begin{table}[h]
\renewcommand{\arraystretch}{1.5}
\begin{center}
\begin{tabular}{||r|c|c|c||}
\hline
$\alpha$ & {lower bound} & $\varepsilon_{X|Y}$ & {upper bound} \\
\hline
1 & 0.4013  & 0.4667 & 0.6061 \\
10 & 0.4640 & 0.4667 & 0.4994 \\
100 & 0.4667 & 0.4667 & 0.4699\\
\hline
\end{tabular}
\vspace*{0.3cm}
\caption{\label{table2} Upper and Lower bounds on $\varepsilon_{X|Y}$
in \eqref{eq1: generalized Fano} and \eqref{eq: UB1 min err} for
Example~\ref{example:45b}.}
\end{center}
\end{table}

As the following example illustrates, the bounds in
\eqref{eq1: generalized Fano} and \eqref{eq: UB1 min err}
are in general not monotonic in $\alpha$.
\begin{example}\label{example:fig2}
Suppose $X$ and $Y$ are binary random variables with joint distribution
\begin{align} \label{eq: joint_dist}
\bigl[P_{XY}(x,y)\bigr]_{(x,y) \in \{0,1\}^2}
= \left( \begin{array}{ccc}
0.1906  & 0.3737 \\
0.4319  & 0.0038
\end{array}
\right).
\end{align}
\begin{figure}[h]
\centerline{\includegraphics[width=10cm]{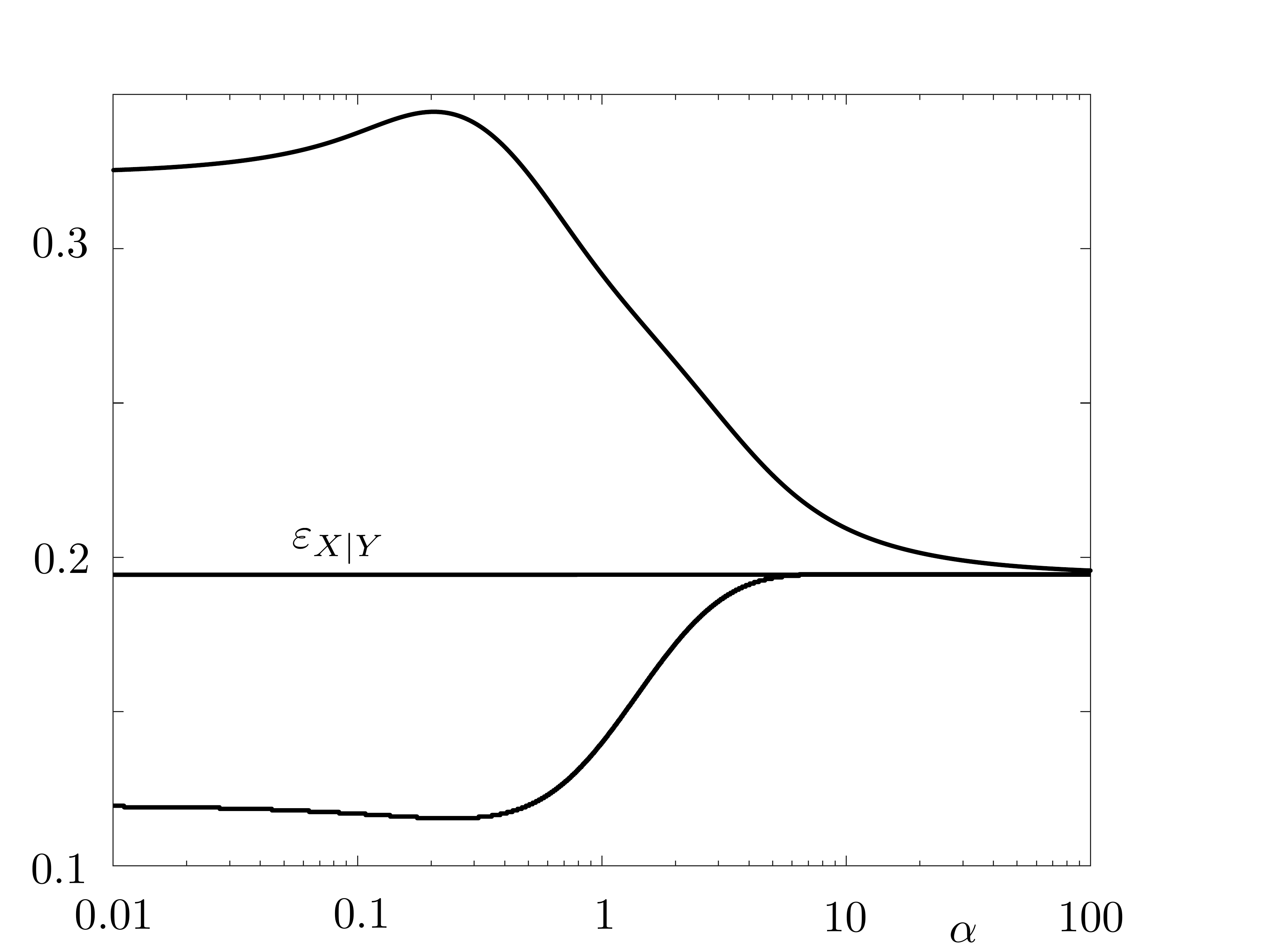}}
\caption{\label{figure: compare_bounds_epsilon_X_given_Y}
Upper (Theorem~\ref{theorem: UB min err}) and lower
(Theorem~\ref{theorem: generalized Fano inequality}) bounds
on $\varepsilon_{X|Y}$ as a function of $\alpha$ in
Example~\ref{example:fig2}.}
\end{figure}
In this case  $\varepsilon_{X|Y} = 0.1944$, and its upper and
lower bounds  are depicted in Figure~\ref{figure: compare_bounds_epsilon_X_given_Y}
as a function of $\alpha$.
\end{example}

\section{Upper Bounds on $\varepsilon_{X|Y}$ Based on Binary Hypothesis Testing}
\label{section: BHT}

In this section, we give explicit upper bounds on $\varepsilon_{X|Y}$
by connecting the $M$-ary hypothesis testing problem for finite $M$
with the associated $\binom{M}{2}$ binary hypothesis testing problems.
The latter bounds generalize the binary hypothesis testing upper bound on
$\varepsilon_{X|Y}$ by Hellman and Raviv \cite[Theorem~1]{HellmanR70}
to any number of hypotheses, including the tightening of the upper
bound by Kanaya and Han \cite{kanaya1995asymptotics}, and Leang and Johnson \cite{LeangJ97}.

Hellman and Raviv \cite[Theorem~1]{HellmanR70} generalized the Bhattacharyya
bound to give an upper bound on the  error probability for \textit{binary}
hypothesis testing in terms of the R\'enyi divergence between both models
$\PZ$ and $\PU$ (with prior probabilities $\mathbb{P} [ \HZ ]$, and
$\mathbb{P} [ \HU ]$ respectively):
\begin{align}
\varepsilon_{X|Y} &\leq \inf_{\alpha\in(0,1)}
\mathbb{P}^\alpha  [ \HU ] \mathbb{P}^{1-\alpha} [ \HZ ]
\exp \bigl( (\alpha -1 ) D_\alpha (\PU \| \PZ) \bigr).
\label{eq:binHR-2}
\end{align}
To generalize \eqref{eq:binHR-2} to any number of hypotheses (including infinity),
it is convenient  to keep the same notation denoting the models by
\begin{align} \label{models}
P_i = P_{Y|X=i}
\end{align}
with prior probabilities  $\mathbb{P} [ \mathsf{H}_{i} ]$ for $i \in \{1, \ldots, M\}$, and let
\begin{align}\label{eq: bar P}
\overbar{\mathbb{P}} [ \mathsf{H}_{i} ] &= 1 - \mathbb{P} [ \mathsf{H}_{i} ], \\[0.1cm]
\overbar{P}_i &= \sum_{j \neq i} \frac{ \mathbb{P} [ \mathsf{H}_{j} ]}{
\overbar{\mathbb{P}} [ \mathsf{H}_{i} ] } \, P_j. \label{eq1b: bar P}
\end{align}

\begin{theorem}  \label{theorem: generalized Hellman-Raviv bound}
\begin{align} \label{eq: generalized Hellman-Raviv bound}
& \varepsilon_{X|Y} \leq \min_{i\neq k} \inf_{\alpha\in(0,1)}
\overbar{\mathbb{P}}^\alpha [ \mathsf{H}_{i} ] \overbar{\mathbb{P}}^{1-\alpha} [ \mathsf{H}_{k} ]
\, \exp \left( (\alpha-1) D_{\alpha} \left( \overbar{P}_i \, \big\| \, \overbar{P}_k \right) \right).
\end{align}
\end{theorem}

\begin{proof}
Let
\begin{align} \label{eq: P_Y}
Y \sim P_Y = \sum_{m=1}^M \mathbb{P} [ \mathsf{H}_{m} ] \, P_m
\end{align}
and, for $m \in \{1, \ldots, M\}$, denote the densities
\begin{align} \label{20171009}
p_m = \frac{\mathrm{d} P_m}{\mathrm{d} P_Y }, \quad
\overbar{p}_m = \frac{\mathrm{d} \overbar{P}_m}{\mathrm{d} P_Y }.
\end{align}
For $\alpha \in (0,1)$ and for all $i,k \in \{1, \ldots, M\}$ with $i \neq k$,
\begin{align} \label{eq1: Hel-Rav}
 \varepsilon_{X|Y} & = 1 - \mathbb{E} \left[ \max_{m \in \{1, \ldots, M\}}
\prob[\mathsf{H}_m | Y] \right] \\[0.1cm]
&= \mathbb{E} \left[ \min_{m \in \{1, \ldots, M\}} \sum_{j \neq m}
\mathbb{P} [ \mathsf{H}_{j} ] \, p_j (Y) \right]
\label{eq2: Hel-Rav}\\[0.2cm]
&=  \mathbb{E} \left[ \min_{m \in \{1, \ldots, M\}} \overbar{\mathbb{P}}
[ \mathsf{H}_{m} ] \, \overbar{p}_m (Y) \right] \label{eq2.5: Hel-Rav} \\[0.2cm]
&\leq \mathbb{E} \left[ \min \left\{  \overbar{\mathbb{P}} [ \mathsf{H}_{i} ]
\, \overbar{p}_i (Y) , \, \overbar{\mathbb{P}} [ \mathsf{H}_{k} ] \, \overbar{p}_k (Y)  \right\} \right]\\[0.2cm]
\label{eq3: Hel-Rav}
&\leq  \mathbb{E} \left[
 \overbar{\mathbb{P}}^{\alpha} [ \mathsf{H}_{i} ] \; \overbar{p}_i^{\alpha} (Y)\;
\overbar{\mathbb{P}}^{1-\alpha} [ \mathsf{H}_{k} ] \; \overbar{p}_k^{1-\alpha} (Y) \right] \\[0.2cm]
\label{eq4: Hel-Rav}
&=  \overbar{\mathbb{P}}^\alpha [ \mathsf{H}_{i} ] \, \overbar{\mathbb{P}}^{1-\alpha}[ \mathsf{H}_{k} ]
\exp \left( (\alpha-1) D_{\alpha} \left( \overbar{P}_i \, \big\| \, \overbar{P}_k \right) \right)
\end{align}
where
\begin{itemize}
\item \eqref{eq2: Hel-Rav} follows from Bayes' rule:
\begin{align}
1- \prob[\mathsf{H}_m | Y]  = \sum_{j \neq m} \prob[\mathsf{H}_j | Y] = \sum_{j \neq m}
\mathbb{P} [ \mathsf{H}_{j} ] \, p_j(Y),
\end{align}
\item \eqref{eq2.5: Hel-Rav} follows from the definition in \eqref{eq1b: bar P},
\item \eqref{eq3: Hel-Rav} follows from $\min \{ t, s\} \leq t^\alpha s^{1-\alpha}$
if $\alpha \in [0,1]$ and $t,s\geq 0$,
\item \eqref{eq4: Hel-Rav} follows from \eqref{eq:RD0}.
\end{itemize}
Minimizing the right side of \eqref{eq4: Hel-Rav} over all $(i,k)$ with $i \neq k$
gives \eqref{eq: generalized Hellman-Raviv bound}.
\end{proof}

We proceed to generalize Theorem~\ref{theorem: generalized Hellman-Raviv bound}
along a different direction. We can obtain upper and lower bounds on the
Bayesian $M$-ary minimal error probability in terms of the minimal error
probabilities of the $\binom{M}{2}$ associated  binary hypothesis testing
problems. For that end, denote the random variable $X$ restricted to
$\{i, j\} \subseteq \set{X}$, $i \neq j$, by $X_{ij}$, namely,
\begin{align}
\mathbb{P}[ X_{ij} = i] &= \frac{\mathbb{P}  [ \mathsf{H}_i ] }{\mathbb{P}
[ \mathsf{H}_i ] +\mathbb{P}  [ \mathsf{H}_j ] } \label{iolai},\\[0.2cm]
\mathbb{P}[ X_{ij} = j] &= \frac{\mathbb{P}  [ \mathsf{H}_j ] }{\mathbb{P}
[ \mathsf{H}_i ] +\mathbb{P}  [ \mathsf{H}_j ] }. \label{jolai}
\end{align}
Guessing the hypothesis without the benefit of observations
yields the error probability
\begin{align}
\varepsilon_{X_{ij} } &=  \min \{ \mathbb{P}[ X_{ij} = i], \mathbb{P}[ X_{ij} = j] \} \\[0.2cm]
& = \frac{\min\{ \mathbb{P}  [ \mathsf{H}_i ] ,
\mathbb{P}  [ \mathsf{H}_j ] \} }{\mathbb{P}  [ \mathsf{H}_i ] +\mathbb{P}  [ \mathsf{H}_j ]} \label{varij}.
\end{align}
The minimal error probability of the binary hypothesis test
\begin{align}
& \mathsf{H}_i: Y_{ij} \sim P_{Y|X_{ij}=i} = P_{Y|X=i} = P_i, \label{eq: Hi} \\[0.1cm]
& \mathsf{H}_j: Y_{ij} \sim P_{Y|X_{ij}=j} = P_{Y|X=j} = P_j \label{eq: Hj}
\end{align}
with the prior probabilities in \eqref{iolai} and \eqref{jolai} is
denoted by $\varepsilon_{X_{ij} | Y_{ij}}$, and it is given by
\begin{align}\label{silito}
\varepsilon_{X_{ij} | Y_{ij} } =
\mathbb{E} \left[ \min\{\mathbb{P}[\mathsf{H}_i |Y_{ij}], \, \mathbb{P}[\mathsf{H}_j |Y_{ij}] \} \right],
\end{align}
with
\begin{align} \label{eq: Y_ij pdf}
Y_{ij} \sim  P_{Y_{ij}} = \mathbb{P}[ X_{ij} = i]  P_i + \mathbb{P}[ X_{ij} = j]  P_j,
\end{align}
and
\begin{align} \label{baijyes}
\mathbb{P}[ \mathsf{H}_i |Y_{ij}=y] =
\mathbb{P}[ X_{ij} = i] \, \frac{\mathrm{d}P_i}{\mathrm{d}P_{Y_{ij}}}(y).
\end{align}
The error probability achieved by the Bayesian $M$-ary MAP decision rule can be upper bounded
in terms of the error probabilities in \eqref{silito} of the $\tfrac12 \, M(M-1)$ binary
hypothesis tests in \eqref{eq: Hi}--\eqref{eq: Hj}, with $1 \leq i < j \leq M$, by means
of the following result.

\begin{theorem}\label{thm:UBMbasedon2HT}
\begin{align}  \label{eq:UBMbasedon2HT-2}
\varepsilon_{X|Y} &\leq \sum_{1 \leq i < j \leq M}
\mathbb{E} \left[ \min\{\mathbb{P}[ \mathsf{H}_i |Y], \mathbb{P}[ \mathsf{H}_j |Y] \} \right]\\[0.1cm]
&= \sum_{1\leq i < j \leq M} \left( \mathbb{P}  [ \mathsf{H}_i ] + \mathbb{P}[ \mathsf{H}_j ]  \right)
\, \varepsilon_{X_{ij} | Y_{ij} }
\label{eq:UBMbasedon2HT}
\end{align}
where $Y \sim P_Y$ with $P_Y$ in \eqref{eq: P_Y},
and $Y_{ij} \sim P_{Y_{ij}}$ in \eqref{eq: Y_ij pdf}.
\end{theorem}
\begin{proof}
Every probability mass function $(q_1, \ldots , q_M )$ satisfies
\begin{align} \label{simple}
1 - \max \{ q_1 , \ldots , q_M \} \leq \sum_{1\leq i < j \leq M}  \min\{ q_i , q_j \}
\end{align}
because, due to symmetry, if without any loss of generality $q_1$ attains the maximum in
the left side of \eqref{simple}, then the partial sum in the right side of \eqref{simple}
over the $M-1$ terms that involve $q_1$ is equal $q_2 + \ldots + q_M = 1 - q_1$ which is
equal to the left side of \eqref{simple}. Hence,
\begin{align}
\varepsilon_{X|Y} &= \int \varepsilon_{X|Y=y} \, \mathrm{d}P_Y(y) \label{ee0} \\[0.1cm]
&= \int \left( 1 - \max \left\{ \mathbb{P}[\mathsf{H}_1 | Y=y], \ldots,
\mathbb{P}[\mathsf{H}_M | Y=y] \right\} \right) \, \mathrm{d}P_Y (y) \label{ee1} \\[0.1cm]
&\leq \int \sum_{1\leq i < j \leq M}  \min\{ \mathbb{P}[\mathsf{H}_i | Y=y],
\, \mathbb{P}[\mathsf{H}_j | Y=y] \} \, \mathrm{d}P_Y (y) \label{ee2} \\[0.1cm]
&= \sum_{1 \leq i < j \leq M}
\mathbb{E} \left[ \min\{\mathbb{P}[ \mathsf{H}_i |Y], \mathbb{P}[ \mathsf{H}_j |Y] \} \right]
\label{ee4}
\end{align}
where \eqref{ee2} follows from \eqref{simple} with $q_i \leftarrow \mathbb{P}[\mathsf{H}_i | Y=y]$
for $i \in \{1, \ldots, M\}$ and $y \in \set{Y}$. This proves \eqref{eq:UBMbasedon2HT-2}.

To show \eqref{eq:UBMbasedon2HT}, note that for all $i, j \in \{1, \ldots, M\}$ with $i \neq j$,
\begin{align}
& \left( \mathbb{P}[ \mathsf{H}_i ] +\mathbb{P}[ \mathsf{H}_j ] \right) \,
\varepsilon_{X_{ij} | Y_{ij} } \nonumber \\[0.1cm]
&= \left( \mathbb{P}[ \mathsf{H}_i ] +\mathbb{P}[ \mathsf{H}_j ] \right) \,
\int \min \bigl\{ \mathbb{P}[\mathsf{H}_i | Y_{ij}=y], \mathbb{P}[\mathsf{H}_j | Y_{ij}=y] \bigr\}
\, \mathrm{d}P_{Y_{ij}}(y) \label{ee5} \\[0.1cm]
&= \left( \mathbb{P}[ \mathsf{H}_i ] +\mathbb{P}[ \mathsf{H}_j ] \right) \,
\int \min \bigl\{ \mathbb{P}[X_{ij}=i] \, \mathrm{d}P_i(y), \,
\mathbb{P}[X_{ij}=j] \, \mathrm{d}P_j(y) \bigr\} \label{ee6} \\[0.1cm]
&= \int \min \bigl\{\mathbb{P}[ \mathsf{H}_i ] \, \mathrm{d}P_i(y), \,
\mathbb{P}[ \mathsf{H}_j ] \, \mathrm{d}P_j (y) \bigr\} \label{ee7} \\[0.1cm]
&= \mathbb{E} \left[ \min\{\mathbb{P}[ \mathsf{H}_i |Y], \mathbb{P}[ \mathsf{H}_j |Y] \} \right] \label{ee8}
\end{align}
where \eqref{ee5} and \eqref{ee6} hold due to \eqref{silito} and \eqref{baijyes}, respectively; \eqref{ee7} follows from
\eqref{iolai} and \eqref{jolai}; \eqref{ee8} follows from the equality
$\mathbb{P}[ \mathsf{H}_i |Y=y] = \mathbb{P}[\mathsf{H}_i] \, \frac{\mathrm{d}P_i}{\mathrm{d}P_Y}(Y=y)$.
Finally, \eqref{eq:UBMbasedon2HT} follows from \eqref{eq:UBMbasedon2HT-2} and \eqref{ee5}--\eqref{ee8}.
\end{proof}

The following immediate consequence of Theorem~\ref{thm:UBMbasedon2HT} is
the bound obtained in
\cite[(3)]{LeangJ97}.
\begin{corollary}
\begin{align}\label{eq2:UBMbasedon2HT}
\varepsilon_{X|Y} \leq
\tfrac12{M(M-1)} \, \max_{i\neq j}\varepsilon_{X_{ij} | Y_{ij} } .
\end{align}
\end{corollary}
We proceed to give upper bounds on $\varepsilon_{X|Y}$ based on
R\'enyi divergence and Chernoff information.
\begin{theorem}\label{thm:UBSV20160616}
\begin{align}
\varepsilon_{X|Y}
&\leq \sum_{1\leq i < j \leq M}  \inf_{\alpha\in(0,1)}
\mathbb{P}^\alpha[ \mathsf{H}_i ] \,
\mathbb{P}^{1-\alpha}[ \mathsf{H}_j ]
\, \exp\bigl( (\alpha -1) D_\alpha(P_i \| P_j)  \bigr)
\label{fra2rum1} \\
&\leq (M-1) \exp \Bigl( - \min_{i\neq j} C ( P_i \| P_j ) \Bigr).
\label{eq:SV20160611}
\end{align}
\end{theorem}
\begin{proof}
Consider the binary hypothesis test in \eqref{eq: Hi}--\eqref{eq: Hj} with its
minimal error probability $\varepsilon_{X_{ij} | Y_{ij} }$.
Substituting the prior probabilities in \eqref{iolai}--\eqref{jolai} into the Hellman-Raviv bound \eqref{eq:binHR-2},
we obtain
\begin{align}
& \left( \mathbb{P}  [ \mathsf{H}_i ] +\mathbb{P}  [ \mathsf{H}_j ]  \right) \varepsilon_{X_{ij} | Y_{ij} } \nonumber \\[0.1cm]
& \leq \left( \mathbb{P}  [ \mathsf{H}_i ] +\mathbb{P}  [ \mathsf{H}_j ]  \right) \inf_{\alpha \in (0,1)}
\Bigl\{ \mathbb{P}^{\alpha}[ X_{ij} = i] \, \mathbb{P}^{1-\alpha}[ X_{ij} = j]
\, \exp \bigl(  ( \alpha - 1) D_\alpha ( P_i \| P_j ) \bigr) \Bigr\} \\[0.1cm]
& = \inf_{\alpha\in(0,1)} \mathbb{P}^\alpha [ \mathsf{H}_i  ] \, \mathbb{P}^{1-\alpha} [ \mathsf{H}_j  ]
\exp \bigl(  ( \alpha - 1) D_\alpha ( P_i \| P_j ) \bigr)
\end{align}
which, upon substitution in \eqref{eq:UBMbasedon2HT}, yields \eqref{fra2rum1}.
To show \eqref{eq:SV20160611}, note that for all $t,s \geq 0$ and $\alpha \in (0,1)$
\begin{align} \label{good}
t^{\alpha} \, s^{1-\alpha} \leq \max\{t,s\} \leq t+s.
\end{align}
Hence, we get
\begin{align}
\label{good1}
\varepsilon_{X|Y}
&\leq \sum_{1\leq i < j \leq M}
(\mathbb{P}[ \mathsf{H}_i ] + \mathbb{P}[ \mathsf{H}_j ])
\, \inf_{\alpha\in(0,1)} \exp\bigl( (\alpha -1) D_\alpha(P_i \| P_j)  \bigr) \\[0.1cm]
\label{good2}
&= \sum_{1\leq i < j \leq M}
(\mathbb{P}[ \mathsf{H}_i ] + \mathbb{P}[ \mathsf{H}_j ])
\, \exp\bigl(-C(P_i \| P_j)\bigr) \\[0.1cm]
\label{good3}
&\leq \exp \Bigl( - \min_{i\neq j} C ( P_i \| P_j ) \Bigr) \,
\sum_{1\leq i < j \leq M}
(\mathbb{P}[ \mathsf{H}_i ] + \mathbb{P}[ \mathsf{H}_j ])
\end{align}
where \eqref{good1} holds due to \eqref{fra2rum1} and \eqref{good}
with $t \leftarrow \mathbb{P}[ \mathsf{H}_i ]$ and
$s \leftarrow \mathbb{P}[ \mathsf{H}_j ]$; \eqref{good2} holds by
the definition in \eqref{eq:def:chernoffinformation} with $P \leftarrow P_i$
and $Q \leftarrow P_j$. The sum in the right
side of \eqref{good3} satisfies
\begin{align}
\sum_{1 \leq i < j \leq M}
(\mathbb{P}[ \mathsf{H}_i ] + \mathbb{P}[ \mathsf{H}_j ])
\label{good4}
&= \tfrac12 \sum_{i \neq j} (\mathbb{P}[ \mathsf{H}_i ]
+ \mathbb{P}[ \mathsf{H}_j ]) \\
\label{good5}
&= \tfrac12 \sum_{i=1}^M \sum_{j=1}^M (\mathbb{P}[ \mathsf{H}_i ]
+ \mathbb{P}[ \mathsf{H}_j ]) - \sum_{i=1}^M \mathbb{P}[H_i] \\
&= M-1. \label{good6}
\end{align}
The result in \eqref{eq:SV20160611} follows from \eqref{good1}--\eqref{good6}.
\end{proof}
\begin{remark}
The bound on $\varepsilon_{X|Y}$ in \eqref{eq:SV20160611} can be improved by
relying on the middle term in \eqref{good}, which then gives
that for $t,s \geq 0$ and $\alpha \in (0,1)$
\begin{align} \label{good-ver2}
t^{\alpha} \, s^{1-\alpha} \leq \max\{t,s\} = \tfrac12 (t+s+|t-s|).
\end{align}
Following the analysis in \eqref{good1}--\eqref{good6}, we get
\begin{align}
\varepsilon_{X|Y} \leq \left( \tfrac{M-1}{2} + \tfrac12 \sum_{1 \leq i < j \leq M}
\bigl| \mathbb{P}[ \mathsf{H}_i ] - \mathbb{P}[ \mathsf{H}_j ] \bigr| \right)
\exp \Bigl( - \min_{i\neq j} C ( P_i \| P_j ) \Bigr),
\end{align}
which yields an improvement over \eqref{eq:SV20160611} by a factor of at most
$\tfrac12$, attained when $X$ is equiprobable.
\end{remark}
\begin{remark}
Using the identity
\begin{align}
(1-\alpha) D_{\alpha}(P\|Q) = \alpha D_{1-\alpha}(Q \| P)
\end{align}
for all $\alpha \in (0,1)$ and probability measures $P$ and $Q$, it is easy to check that
Theorems~\ref{theorem: generalized Hellman-Raviv bound}
and~\ref{thm:UBSV20160616} result in the same bound for binary hypothesis testing.
\end{remark}
\begin{remark}
Kanaya and Han \cite{kanaya1995asymptotics} and Leang and Johnson \cite{LeangJ97} give
a weaker bound with $\frac12 M (M-1)$ in lieu of $M-1$ in \eqref{eq:SV20160611}.
\end{remark}
\begin{remark}
In the absence of observations,
the bound in  \eqref{eq: generalized Hellman-Raviv bound} is tight
since $\varepsilon_{X|Y} = 1 - \max_{i} \mathbb{P}  [ \mathsf{H}_i ].$
On the other hand, \eqref{eq2:UBMbasedon2HT} and \eqref{fra2rum1} need not
be tight in that case, and in fact may be strictly larger than~1.
\end{remark}

Numerical experimentation with Examples~\ref{example:45b}--\ref{example:fig2}
and others shows that, in general, the Arimoto-R\'enyi conditional entropy bounds
are tighter than \eqref{eq: generalized Hellman-Raviv bound} and \eqref{fra2rum1}.

\section{Arimoto-R\'{e}nyi Conditional Entropy Averaged over Codebook Ensembles}
\label{section: random coding}
In this section we consider the channel coding setup with a code ensemble
$\mathcal{C}$, over which we are interested in averaging the Arimoto-R\'{e}nyi
conditional entropy of the channel input given the channel output. We denote such
averaged quantity by $\expectation_{\set{C}}\bigl[H_{\alpha}(X^n | Y^n)\bigr]$
where $X^n = (X_1, \ldots, X_n)$ and $Y^n = (Y_1, \ldots, Y_n)$.
Some motivation for this study arises from the fact that the normalized equivocation
$\frac1n H(X^n|Y^n)$ as a reliability measure was used by Shannon \cite{CES48}
in his proof that reliable communication is impossible at rates above channel capacity;
furthermore, the asymptotic convergence to zero of the equivocation $H(X^n | Y^n)$
at rates below capacity was studied by Feinstein \cite[Section~3]{Feinstein54}.

We can capitalize on the convexity of $d_\alpha (\cdot \| q) $ for $\alpha \in [0,1]$
(see \cite[Theorem~11]{ErvenH14} and \cite[Section~6.A]{ErvenH14}) to claim that
in view of Theorem~\ref{theorem: generalized Fano inequality}, for any ensemble
$\set{C}$ of size-$M$ codes and blocklength $n$,
\begin{align} \label{gf1}
\expectation_{\set{C}}\bigl[H_{\alpha}(X^n | Y^n)\bigr]
\leq \log M - d_{\alpha}\bigl( \overbar{\varepsilon}\, \| 1-\tfrac1M \bigr)
\end{align}
where $\alpha \in [0,1]$ and $\overbar{\varepsilon}$ is the error probability of the
maximum-likelihood decoder averaged over the code ensemble, or an upper bound
thereof which does not exceed $1-\frac1M$ (for such upper bounds see, e.g.,
\cite[Chapters~1--4]{FnT06}).
Analogously, in view of Theorem~\ref{theorem: generalized Fano-Renyi - list decoding}
for list decoding with a fixed list size $L$, the same convexity argument implies that
\begin{align} \label{gf2}
\expectation_{\set{C}}\bigl[H_{\alpha}(X^n | Y^n)\bigr]  \leq \log M -
d_{\alpha}\bigl( \overbar{P_{\set{L}}}\, \| 1-\tfrac{L}{M} \bigr)
\end{align}
where $\overbar{P_{\set{L}}}$ denotes the minimum list decoding error probability averaged
over the ensemble, or an upper bound thereof which does not exceed $1-\frac{L}{M}$ (for
such bounds see, e.g., \cite{Forney68}, \cite[Problem~5.20]{Gallager_book1968},
\cite[Section~5]{HofSS10} and \cite{Merhav14}).

In the remainder of this section, we consider the discrete memoryless channel (DMC) model
\begin{align}
P_{Y^n|X^n} (y^n |x^n) = \prod_{i=1}^n P_{Y|X}(y_i | x_i),
\end{align}
an ensemble $\set{C}$ of size-$M$ codes and blocklength $n$ such that
the $M$ messages are assigned independent codewords, drawn i.i.d. with
per-letter distribution $P_X$ on the input alphabet
\begin{align}
P_{X^n} (x^n) = \prod_{i=1}^n P_X(x_i).
\end{align}
In this setting, Feder and Merhav show in \cite{FederM94} the following result
for the Shannon conditional entropy.
\begin{theorem} \cite[Theorem 3]{FederM94}
\label{theorem: fm}
For a DMC with transition
probability matrix $P_{Y|X}$, the conditional entropy of the transmitted codeword given the
channel output, averaged over a random coding selection with per-letter distribution $P_X$
such that $I(P_X,P_{Y|X})>0$, is
bounded by
\begin{align}\label{mf1}
\expectation_{\set{C}}\bigl[H (X^n | Y^n)\bigr]
&\leq \inf_{\rho \in (0,1]} \left( 1 + \frac1\rho \right)
\exp \left( - n \rho \left( I_{\frac{1}{1 +\rho}} (P_X, P_{Y|X}) -  R\right) \right)
\log \mathrm{e} \\[0.1cm]
&\leq  \left(1 + \frac1{\rho^\ast (R,P_X)} \right)
\exp \bigl( -n E_{\mathrm{r}}(R, P_X ) \bigr) \log \mathrm{e}
\label{mf2}
\end{align}
with
\begin{align}
R = \tfrac1n \log M \leq I(P_X, P_{Y|X}),
\end{align}
where $E_{\mathrm{r}}$ in the right side of \eqref{mf2} denotes the random-coding error exponent, given by (recall \eqref{eq: E0_a})
\begin{align}\label{erq}
E_{\mathrm{r}}(R, P_X ) = \max_{\rho\in[0,1]} \rho\, \left( I_{\frac{1}{1 +\rho}} (P_X, P_{Y|X}) -  R \right),
\end{align}
and the argument that maximizes \eqref{erq} is denoted in \eqref{mf2} by $\rho^\ast (R,P_X)$.
\end{theorem}
\begin{remark}
Since $H_{\alpha}(X^n | Y^n)$ is monotonically decreasing in $\alpha$ (see
Proposition~\ref{prop1: mon}), the upper bounds in \eqref{mf1}--\eqref{mf2}
also apply to $\expectation_{\set{C}}\bigl[H_\alpha (X^n | Y^n)\bigr]$ for
$\alpha \in [1,\infty]$.
\end{remark}

\par
We are now ready to state and prove the main result in this section.
\begin{theorem}  \label{theorem: random coding}
The following results hold under the setting in Theorem~\ref{theorem: fm}:
\begin{enumerate}
\item For all $\alpha >0$, and rates $R$ below the channel capacity $C$,
\begin{align} \label{sv1}
\underset{n \to \infty}{\lim\sup} \, -\frac1n \,
\log \expectation_{\set{C}}\bigl[H_{\alpha}(X^n | Y^n)\bigr]
\leq E_{\text{sp}}(R),
\end{align}
where $E_{\text{sp}}(\cdot)$ denotes the sphere-packing error exponent
\begin{align}\label{espq}
E_{\text{sp}}(R) = \sup_{\rho \geq 0} \, \rho \,
\left( \max_{Q_X}  I_{\frac{1}{1 +\rho}} (Q_X, P_{Y|X}) -  R \right)
\end{align}
with the maximization in the right side of \eqref{espq} over all
single-letter distributions $Q_X$ defined on the input alphabet.
\item For all $\alpha \in (0,1)$,
\begin{align}  \label{sv2}
\underset{n \to \infty}{\lim\inf} \, -\frac1n \, \log \expectation_{\set{C}}\bigl[H_{\alpha}(X^n | Y^n)\bigr]
\geq \alpha E_{\mathrm{r}}(R,P_X) - (1-\alpha)R,
\end{align}
provided that
\begin{align} \label{eq: R_alpha}
R < R_{\alpha}(P_X, P_{Y|X})
\end{align}
where $R_{\alpha}(P_X, P_{Y|X})$ is the unique solution $r \in (0, I (P_X, P_{Y|X}))$ to
\begin{align}  \label{Eq. for R_alpha}
E_{\mathrm{r}}(r, P_X) = \left( \frac1\alpha - 1 \right) r.
\end{align}
\item  \label{item: R_alpha prop.}
The rate $R_{\alpha}(P_X, P_{Y|X})$ is monotonically increasing and continuous in $\alpha \in (0,1)$, and
\begin{align}
\label{eq2: limit R_alpha}
\lim_{\alpha \downarrow 0} R_{\alpha}(P_X, P_{Y|X}) &= 0,\\
\label{eq1: limit R_alpha}
\lim_{\alpha \uparrow 1} R_{\alpha}(P_X, P_{Y|X}) &= I (P_X, P_{Y|X}).
\end{align}
\end{enumerate}
\end{theorem}

\begin{proof}
In proving the first two items we actually give an asymptotic lower bound and a non-asymptotic upper bound,
which yield the respective results upon taking $\underset{n \to \infty}{\lim\sup} -\frac1n \log ( \cdot )$
and $\underset{n \to \infty}{\lim\inf} -\frac1n \log ( \cdot )$, respectively.
\begin{enumerate}
\item
From \eqref{eq: Barcelona} and the sphere-packing lower bound \cite[Theorem~2]{SGB} on the decoding
error probability, the following inequality holds for all rates $R \in [0, C)$ and $\alpha \in (0,1) \cup (1, \infty)$,
\begin{align} \label{eq: LB optimal code}
\expectation_{\set{C}}\bigl[H_{\alpha}(X^n | Y^n)\bigr] \geq
\frac{\alpha}{1 - \alpha} \;
\log \Bigl( 1 + \bigl( 2^\frac1{\alpha} -2 \bigr) \exp\bigl(-n [E_{\text{sp}}(R-o_1(n)) + o_2(n)] \bigr) \Bigr)
\end{align}
where $o_1(n)$ and $o_2(n)$ tend to zero as $n \to \infty$. Note that $\eqref{eq: Barcelona}$ holds when the
error probability of the optimal code lies in $[0, \tfrac12]$, but this certainly holds for all $R<C$ and
sufficiently large $n$ (since this error probability decays exponentially to zero).

\item
Note that $E_{\mathrm{r}}(R,P_X) > 0$ since, for $\alpha \in (0,1)$,
\begin{align} \label{20170911-1}
R < R_{\alpha}(P_X, P_{Y|X}) < I(P_X, P_{Y|X}),
\end{align}
where the left inequality in \eqref{20170911-1} is the condition in \eqref{eq: R_alpha}, and
the right inequality in \eqref{20170911-1} holds by the definition of $R_{\alpha}(P_X, P_{Y|X})$
in \eqref{Eq. for R_alpha} (see its equivalent form in \eqref{eq2: R_alpha}).
In view of the random coding bound \cite{Gallager_IT1965}, let
\begin{align}
\overbar{\varepsilon}_n = \overbar{\varepsilon}_n(R,P_X)
\triangleq \min \bigl\{ \exp\bigl(-n E_{\mathrm{r}}(R,P_X)\bigr), \, 1-\exp(-nR) \bigr\}
\label{eq: eps_n}
\end{align}
be an upper bound on the decoding error probability of the code ensemble $\mathcal{C}$, which
does not exceed $1-\frac1M$ with $M=\exp(nR)$. From \eqref{gf1}, for all $\alpha \in (0,1)$,
\begin{align}
& \expectation_{\set{C}}\bigl[H_{\alpha}(X^n | Y^n)\bigr] \nonumber \\[0.1cm]
\label{random coding 0}
&\leq nR - d_{\alpha}\bigl( \overbar{\varepsilon}_n \,  \| \, 1 - \exp(-nR) \bigr) \\[0.1cm]
&= \tfrac1{1-\alpha} \, \log \Bigl( (1-\overbar{\varepsilon}_n)^{\alpha} +
\bigl(\exp(nR)-1\bigr)^{1-\alpha} \, \overbar{\varepsilon}_n^{\, \alpha} \Bigr)
\label{20170911-2} \\[0.1cm]
&\leq \tfrac1{1-\alpha} \, \log \Bigl( (1-\overbar{\varepsilon}_n)^{\alpha} +
\exp\bigl(n(1-\alpha)R\bigr) \, \overbar{\varepsilon}_n^{\, \alpha} \Bigr) \\[0.1cm]
&\leq \tfrac1{1-\alpha} \, \log \biggl( \Bigl(1 - \exp\bigl(-n E_{\mathrm{r}}(R,P_X)\bigr)
\Bigr)^{\alpha} + \exp\Bigl(-n \bigl[ \alpha E_{\mathrm{r}}(R,P_X) - (1-\alpha)R \bigr] \Bigr) \biggr),
\label{random coding 1}
\end{align}
where equality \eqref{20170911-2} follows from \eqref{useful equality} by setting
$(\theta, s, t) = (\exp(nR), \, \exp(nR)-1, \, \overbar{\varepsilon}_n)$;
the inequality in \eqref{random coding 1} holds with equality for all sufficiently large $n$ such that
$\overbar{\varepsilon}_n = \exp\bigl(-n E_{\mathrm{r}}(R,P_X)\bigr)$ (see \eqref{eq: eps_n}). The
result in \eqref{sv2} follows consequently by taking
$\underset{n \to \infty}{\lim\inf} -\frac1n \log ( \cdot )$ of
\eqref{random coding 0}--\eqref{random coding 1}. To that end, note that the second exponent in
the right side of \eqref{random coding 1} decays to zero whenever
\begin{align}
R &<R_{\alpha}(P_X, P_{Y|X}) \\[0.1cm]
\label{eq2: R_alpha}
&= \sup \Bigl\{ r \in [0, I (P_X, P_{Y|X})): E_{\mathrm{r}}(r, P_X) > \left( \tfrac1\alpha - 1 \right) r \Bigr\}.
\end{align}
It can be verified that the representation of $R_{\alpha}(P_X, P_{Y|X})$ in \eqref{eq2: R_alpha} is indeed
equivalent to the way it is defined in \eqref{Eq. for R_alpha} because
$E_{\mathrm{r}}(r, P_X)$ is monotonically decreasing and continuous in $r$ (see \cite[(5.6.32)]{Gallager_book1968}),
the right side of \eqref{Eq. for R_alpha} is strictly monotonically increasing and continuous in $r$ for $\alpha \in (0,1)$,
and the left and right sides of \eqref{Eq. for R_alpha} vanish, respectively, at $r = I (P_X, P_{Y|X})$
and $r=0$. The equivalent representation of $R_{\alpha}(P_X, P_{Y|X})$ in \eqref{eq2: R_alpha} yields the
uniqueness of the solution $r = R_{\alpha}(P_X, P_{Y|X})$ of \eqref{Eq. for R_alpha} in the interval $(0, I (P_X, P_{Y|X}))$
for all $\alpha \in (0,1)$, and it justifies in particular the right inequality in \eqref{20170911-1}.

\item
From  \eqref{Eq. for R_alpha}, for $\alpha \in (0,1)$, $r = R_{\alpha}(P_X, P_{Y|X})$ is
the unique solution $r \in (0, I (P_X, P_{Y|X}))$ to
\begin{align}  \label{Eq2. for R_alpha}
\tfrac1r \, E_{\mathrm{r}}(r, P_X) = \tfrac1\alpha - 1.
\end{align}
Since $E_{\mathrm{r}}(r, P_X)$ is positive and monotonically decreasing in $r$ on the interval $(0, I (P_X, P_{Y|X}))$,
so is the left side of \eqref{Eq2. for R_alpha} as a function of $r$.
Since also the right side of \eqref{Eq2. for R_alpha} is positive and monotonically decreasing in $\alpha \in (0,1)$,
it follows from \eqref{Eq2. for R_alpha} that $r = R_{\alpha}(P_X, P_{Y|X})$ is monotonically
increasing in $\alpha$. The continuity of $R_{\alpha}(P_X, P_{Y|X})$ in $\alpha$ on the interval $(0,1)$
holds due to \eqref{Eq2. for R_alpha} and the continuity of $E_{\mathrm{r}}(\cdot, P_X)$ (see \cite[p.~143]{Gallager_book1968}).
The limits in \eqref{eq2: limit R_alpha} and \eqref{eq1: limit R_alpha} follow from \eqref{Eq. for R_alpha}:
\eqref{eq2: limit R_alpha} holds since we have
$\underset{r \downarrow 0}{\lim} \, E_{\mathrm{r}}(r, P_X) = E_{\mathrm{r}}(0, P_X) \in (0, \infty)$, whereas
$\underset{\alpha \downarrow 0}{\lim} \left(\frac1\alpha-1\right) = +\infty$;
\eqref{eq1: limit R_alpha} holds since $E_{\mathrm{r}}(R,P_X)$ is equal to zero at
$R=I (P_X, P_{Y|X})$, and it is positive if $R<I (P_X, P_{Y|X})$ (see \cite[p.~142]{Gallager_book1968}).
\end{enumerate}
\end{proof}

\begin{remark}
For $\alpha=1$, Theorem~\ref{theorem: random coding} strengthens the result in
Theorem~\ref{theorem: fm} by also giving the converse in Item~1).
This result enables to conclude that if $P_X$ is the input distribution which
maximizes the random-coding error exponent
\begin{align}
E_{\mathrm{r}}(R) &= \max_{Q} E_{\mathrm{r}}(R,Q), \label{max over P_X}
\end{align}
then
\begin{align}
\lim_{n \to \infty} -\frac1n \, \log
\expectation_{\set{C}} \bigl[H(X^n | Y^n)\bigr] & = E_{\mathrm{r}}(R)
\end{align}
at all rates between the critical rate $R_{\mathrm{c}}$ and channel capacity $C$ of the DMC,
since for all $R \in [R_{\mathrm{c}}, C]$ \cite[Section 5.8]{Gallager_book1968}
\begin{align} \label{Rc to C}
E_{\mathrm{r}}(R) = E_{\text{sp}}(R)
\end{align}
where $E_{\text{sp}}(\cdot)$ is given in \eqref{espq}.
\end{remark}

\begin{remark}
With respect to \eqref{Rc to C}, recall that although the random-coding error exponent is tight
for the average code \cite{Gallager_IT1973}, it coincides with the sphere-packing error exponent
(for the optimal code) at the high-rate region between the critical rate and channel capacity.
As shown by Gallager in \cite{Gallager_IT1973}, the fact that the random coding error
exponent is not tight for optimal codes at rates below the critical rate of the DMC stems from the
poor performance of the bad codes in this ensemble rather than a weakness of the bounding technique
in \cite{Gallager_IT1965}.
\end{remark}

\begin{remark}
The result in Theorem~\ref{theorem: random coding} can be extended to list decoding
with a fixed size $L$ by relying on \eqref{gf2}, and the upper bound on the list
decoding error probability in \cite[Problem~5.20]{Gallager_book1968} where
\begin{align}
\overbar{P_{\set{L}}} = \min\left\{ \exp\bigl(-n E_L(R,P_X)\bigr), 1 - \frac{L}{M} \right\}
\end{align}
with $M = \exp(nR)$, and
\begin{align}\label{elrq}
E_L ( R, P_X ) = \max_{\rho\in[0,L]} \rho\, \left( I_{\frac{1}{1 +\rho}} (P_X, P_{Y|X}) -  R \right).
\end{align}
This result can be further generalized to structured code ensembles by relying
on upper bounds on the list decoding error probability in \cite[Section~5]{HofSS10}.
\end{remark}

The following result is obtained by applying Theorem~\ref{theorem: random coding} in the
setup of communication over a memoryless binary-input output-symmetric channel with a symmetric
input distribution. The following rates, specialized to the considered setup, are required for
the presentation of our next result.
\begin{enumerate}[1)]
\item
The cutoff rate is given by
\begin{align} \label{rate1a}
\RZ &= E_0(1, P_X^{\ast}) \\
\label{rate1b}
&= 1 - \log\bigl(1+\rho_{Y|X}\bigr) ~ \mbox{bits}
\end{align}
where in the right side of \eqref{rate1a}, $E_0(\rho, P_X^{\ast})$ is given in \eqref{eq0: E0_b} with
the symmetric binary input distribution $P_X^{\ast} = \bigl[\tfrac12 ~ \tfrac12\bigr]$,
and $\rho_{Y|X} \in [0,1]$ in the right side of \eqref{rate1b} denotes the Bhattacharyya constant given by
\begin{align} \label{Bhat}
\rho_{Y|X} = \sum_{y \in \set{Y}} \sqrt{ P_{Y|X}(y|0) \, P_{Y|X}(y|1) }
\end{align}
with an integral replacing the sum in the right side of \eqref{Bhat} if $\set{Y}$ is a non-discrete set.
\item
The critical rate is given by
\begin{align} \label{rate2}
R_{\mathrm{c}} = E_0'(1, P_X^{\ast})
\end{align}
where the differentiation of $E_0$ in the right side of \eqref{rate2} is with respect to $\rho$ at $\rho=1$.
\item
The channel capacity is given by
\begin{align}  \label{rate3}
C = I(P_X^{\ast}, P_{Y|X}).
\end{align}
\end{enumerate}
In general, the random coding error exponent is given by (see \cite[Section~5.8]{Gallager_book1968})
\begin{align} \label{eq: MBIOS}
E_{\mathrm{r}}(R) =
\begin{dcases}
\RZ - R, \quad & 0 \leq R \leq R_{\mathrm{c}} \\
E_{\mathrm{sp}}(R), \quad & R_{\mathrm{c}} \leq R \leq C,
\end{dcases}
\end{align}
and
\begin{align} \label{ranking}
0 \leq R_{\mathrm{c}} \leq \RZ \leq C.
\end{align}
Consequently, $E_{\mathrm{r}}(R)$ in \eqref{eq: MBIOS} is composed of two parts: a straight line at
rates $R \in [0, R_{\mathrm{c}}]$, starting from $\RZ$ at zero rate with slope $-1$, and
the sphere-packing error exponent at rates $R \in [R_{\mathrm{c}}, C]$.

\begin{theorem} \label{theorem: Ralpha - MBIOS}
Let $P_{Y|X}$ be the transition probability matrix of a memoryless binary-input
output-symmetric channel, and let $P_X^{\ast} = \bigl[\tfrac12 ~ \tfrac12\bigr]$.
Let $R_{\mathrm{c}}$,
$\RZ$, and $C$ denote the critical and cutoff rates and the channel capacity,
respectively, and let\footnote{In general $\alpha_{\mathrm{c}} \in [0,1]$
(see \eqref{ranking}); the cases $\alpha_{\mathrm{c}}=0$ and $\alpha_{\mathrm{c}}=1$
imply that Theorem~\ref{theorem: Ralpha - MBIOS}-\ref{case1 - Ralpha}) or~\ref{case2 - Ralpha})
are not applicable, respectively.}
\begin{align} \label{critical-over-cutoff}
\alpha_{\mathrm{c}} = \frac{R_{\mathrm{c}}}{\RZ} \in (0,1).
\end{align}
The rate $R_{\alpha} = R_{\alpha}(P_X^{\ast}, P_{Y|X})$, as introduced in Theorem~\ref{theorem: random coding}--2)
with the symmetric input distribution $P_X^{\ast}$, can be expressed as follows:
\begin{enumerate}[a)]
\item \label{case1 - Ralpha}
for $\alpha \in (0, \alpha_{\mathrm{c}}]$,
\begin{align} \label{eq1: R_alpha - MBIOS}
R_{\alpha} = \alpha \RZ;
\end{align}
\item \label{case2 - Ralpha}
for $\alpha \in (\alpha_{\mathrm{c}}, 1)$,
$R_{\alpha} \in (R_{\mathrm{c}}, C)$ is the solution to
\begin{align}  \label{eq2: R_alpha - MBIOS}
E_{\mathrm{sp}}(r) = \left( \frac1\alpha - 1 \right) r;
\end{align}
$R_{\alpha}$ is continuous and monotonically increasing
in $\alpha \in [\alpha_{\mathrm{c}}, 1)$ from $R_{\mathrm{c}}$ to $C$.
\end{enumerate}
\end{theorem}

\begin{proof}
For every memoryless binary-input output-symmetric channel, the symmetric input distribution
$P_X^{\ast}$ achieves the maximum of the error exponent $E_{\mathrm{r}}(R)$ in the
right side of \eqref{max over P_X}. Consequently, \eqref{eq1: R_alpha - MBIOS} and
\eqref{eq2: R_alpha - MBIOS} readily follow from \eqref{Eq. for R_alpha}, \eqref{Rc to C},
\eqref{eq: MBIOS} and \eqref{critical-over-cutoff}. Due to Theorem~\ref{theorem: random coding}-3),
$R_{\alpha}$ is monotonically increasing and continuous in $\alpha \in (0,1)$. From
\eqref{critical-over-cutoff} and \eqref{eq1: R_alpha - MBIOS}
\begin{align}
R_{\alpha_{\mathrm{c}}} = R_{\mathrm{c}},
\end{align}
and \eqref{eq1: limit R_alpha} and \eqref{rate3} yield
\begin{align}
\lim_{\alpha \uparrow 1} R_{\alpha} &= I (P_X^{\ast}, P_{Y|X}) = C.
\end{align}
Consequently, $R_{\alpha}$ is monotonically increasing in $\alpha \in [\alpha_{\mathrm{c}}, 1)$ from $R_{\mathrm{c}}$ to $C$.
\end{proof}

\begin{example} \label{example: bsc}
Let $P_{Y|X} (0|0) = P_{Y|X} (1|1) = 1 -\delta$, $P_{Y|X} (0|1) = P_{Y|X} (1|0) = \delta$,
and $P_X(0)=P_X(1)=\tfrac12$. In this case, it is convenient to express all rates in bits;
in particular, it follows from \eqref{rate1a}--\eqref{rate3} that the
cutoff rate, critical rate and channel capacity are given, respectively, by (see, e.g.,
\cite[p.~146]{Gallager_book1968})
\begin{align}
\label{bsc R0}
\RZ &= 1 - \log\bigl(1+\sqrt{4\delta(1- \delta)}\bigr), \\[0.1cm]
\label{eq: bsc Rc}
R_{\mathrm{c}} &= 1 - h\left(\frac{\sqrt{\delta}}{\sqrt{\delta}+\sqrt{1- \delta}}\right), \\[0.1cm]
\label{bcs capacity}
C &=I (P_X, P_{Y|X}) = 1 - h(\delta).
\end{align}
The sphere-packing error exponent for the binary symmetric channel is given by (see, e.g.,
\cite[(5.8.26)--(5.8.27)]{Gallager_book1968})
\begin{align} \label{eq: BSC}
E_{\mathrm{sp}}(R) = d\bigl(\delta_{\mathrm{GV}}(R) \, \| \, \delta\bigr)
\end{align}
where the normalized Gilbert-Varshamov distance (see, e.g., \cite[Theorem~4.10]{Roth}) is denoted by
\begin{align} \label{gv}
\delta_{\mathrm{GV}}(R) = h^{-1}(1 - R)
\end{align}
with $h^{-1} \colon [0, \log 2] \to \bigl[0, \tfrac12\bigr]$ standing for the inverse of the binary entropy function.

In view of Theorem~\ref{theorem: Ralpha - MBIOS}, Figure \ref{figure: R_alpha} shows
$R_{\alpha}$ with $\alpha \in (0,1)$ for the case where the binary symmetric channel
has capacity $\frac12$ bit per channel use, namely, $\delta=h^{-1}\bigl(\tfrac12 \bigr)=0.110$,
in which case $\alpha_{\mathrm{c}} = 0.5791$ and $R_{\alpha_{\mathrm{c}}} = R_{\mathrm{c}} = 0.1731$
bits per channel use (see \eqref{eq: bsc Rc}).
\begin{figure}[h]
\centerline{\includegraphics[width=10cm]{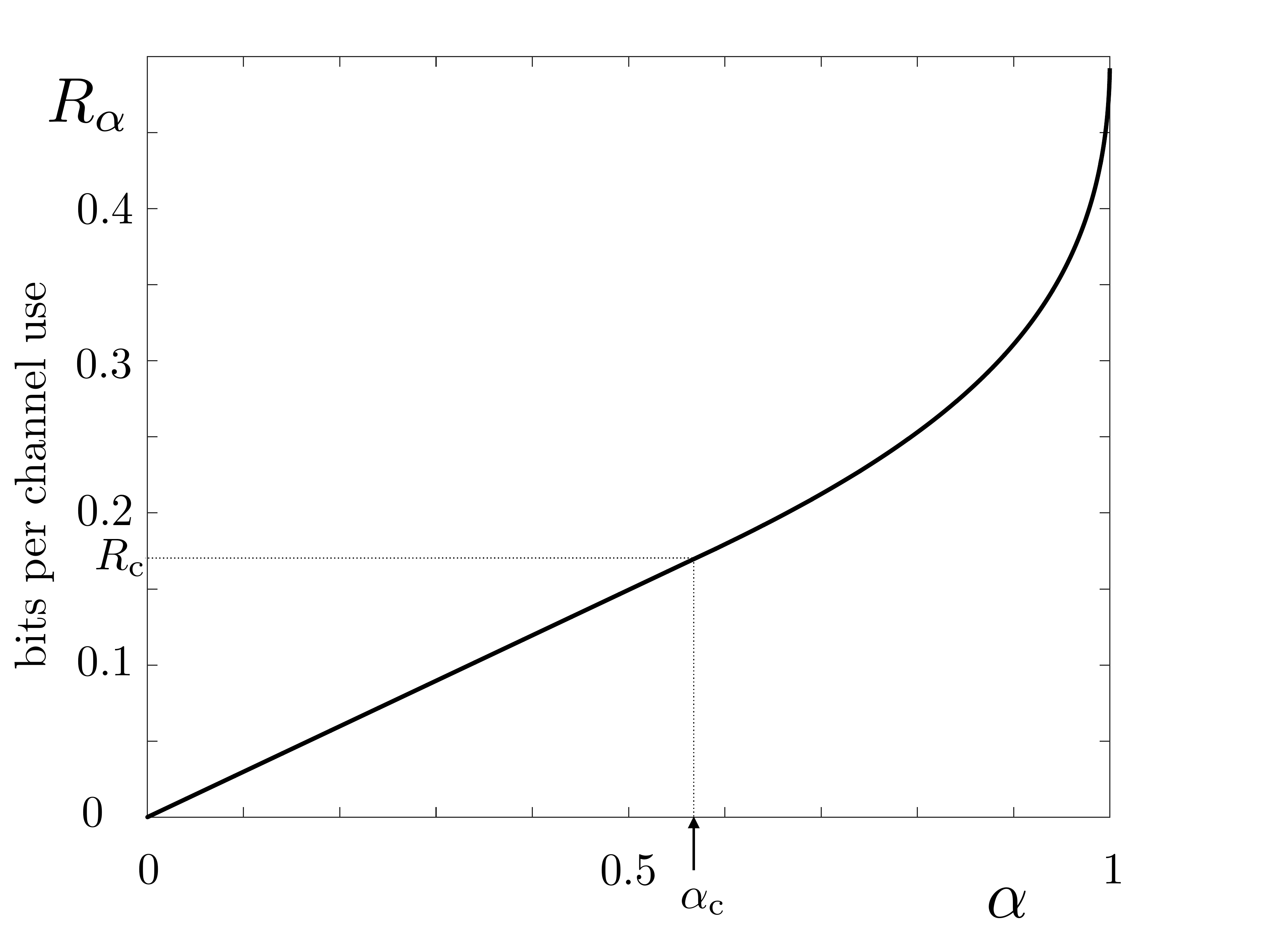}}
\caption{\label{figure: R_alpha}
The rate $R_{\alpha}$ for $\alpha \in (0,1)$ for a binary symmetric channel
with crossover probability $\delta=0.110$.}
\end{figure}
\end{example}

\section{Conclusions}
\label{sec: conclusions}
The interplay between information-theoretic measures and the analysis of hypothesis
testing has been a fruitful area of research. Along these lines, we have shown new
bounds on the minimum Bayesian error probability $\varepsilon_{X|Y}$ of arbitrary
$M$-ary hypothesis testing problems as a function of information measures. In particular,
our major focus has been the Arimoto-R\'enyi conditional entropy of the hypothesis
index given the observation. We have seen how changing the conventional form of Fano's
inequality from
\begin{align}
H(X | Y)
& \leq h(\varepsilon_{X|Y}) + \varepsilon_{X|Y} \log(M-1) \\
& = \log M - d\bigl( \varepsilon_{X|Y} \| 1-\tfrac1M \bigr)\label{alt@fano}
\end{align}
to the right side of \eqref{alt@fano}, where $d(\cdot\|\cdot)$ is the binary relative entropy,
allows a natural generalization where the
Arimoto-R\'enyi conditional entropy of an arbitrary positive order $\alpha$ is upper bounded by
\begin{align} \label{eq1: generalized Fano!alt}
H_{\alpha}(X | Y) \leq \log M -
d_{\alpha}\bigl( \varepsilon_{X|Y} \| 1-\tfrac1M \bigr)
\end{align}
with $d_{\alpha}(\cdot\|\cdot)$ denoting the binary R\'enyi divergence.
\par
Likewise, thanks to the Schur-concavity of R\'enyi entropy, we obtain a lower bound
on $H_{\alpha}(X | Y)$ in terms of $\varepsilon_{X|Y}$, which holds even if $M = \infty$. Again, we are able to recover
existing bounds by letting $\alpha \to 1$.
\par
In addition to the aforementioned bounds on $H_{\alpha}(X | Y)$ as a function of $ \varepsilon_{X|Y}$, it is also of
interest to give explicit lower and upper bounds on $ \varepsilon_{X|Y}$ as a function of $H_{\alpha}(X | Y)$. As
$\alpha \to \infty$, these bounds converge to $ \varepsilon_{X|Y}$.
\par
The list decoding setting, in which the hypothesis tester is allowed to output a subset of given cardinality and an error
occurs if the true hypothesis is not in the list, has considerable interest in information theory. We have shown that
our techniques readily generalize to that setting and we have found generalizations of all the $H_{\alpha}(X | Y)$--$\varepsilon_{X|Y}$
bounds to the list decoding setting.
\par
We have also explored some facets of the role of binary hypothesis testing in analyzing $M$-ary hypothesis testing problems,
and have shown new bounds in terms of R\'enyi divergence.
\par
As an illustration of the application of the $H_{\alpha}(X | Y)$--$\varepsilon_{X|Y}$ bounds, we have analyzed
the exponentially vanishing decay of the Arimoto-R\'enyi conditional entropy of the transmitted
codeword given the channel output for discrete memoryless channels and random coding ensembles.

\appendices
\section{Proof of Proposition~\ref{prop1: mon}} \label{app: mon}
We prove that $H_{\alpha}(X|Y) \leq H_{\beta}(X|Y)$ in three different cases:
\begin{itemize}
\item
$1 < \beta < \alpha $
\item
$0 < \beta < \alpha < 1$
\item
$\beta < \alpha < 0$.
\end{itemize}
Proposition~\ref{prop1: mon} then follows by transitivity, and the continuous extension of
$H_{\alpha}(X|Y)$ at $\alpha=0$ and $\alpha=1$. The following notation is handy.
\begin{align}
\label{eq: rho_2}
\theta &= \frac{\beta-1}{\alpha-1} >0,\\
\label{eq: rho_1}
\rho &=  \frac{\alpha \theta }{\beta} > 0.
\end{align}

{\em Case 1}: If $1 < \beta < \alpha $, then $(\rho, \theta) \in (0,1)^2$, and
\begin{align}
& \exp\bigl(-H_{\alpha}(X|Y)\bigr) \nonumber \\[0.1cm]
\label{eq: case1-a}
&= \mathbb{E}^{\frac{\beta \rho}{{\beta-1}}}
\left[ \left( \sum_{x \in \set{X}} P_{X|Y}^{\alpha}(x|Y) \right)^{\frac1{\alpha}} \right] \\[0.1cm]
\label{eq: case1-b}
&\geq\mathbb{E}^{\frac{\beta}{{\beta-1}}} \left[ \left( \sum_{x \in \set{X}}
P_{X|Y}^{\alpha}(x|Y) \right)^{\frac{\rho}{\alpha}} \right] \\[0.1cm]
\label{eq: case1-c}
&=\mathbb{E}^{\frac{\beta}{{\beta-1}}} \left[ \left( \sum_{x \in \set{X}}
P_{X|Y}(x|y) \, P_{X|Y}^{\alpha-1}(x|Y) \right)^{\frac{\theta}{\beta}} \right] \\[0.1cm]
\label{eq: case1-d}
&\geq\mathbb{E}^{\frac{\beta}{{\beta-1}}} \left[ \left( \sum_{x \in \set{X}}
P_{X|Y}(x|y) \, P_{X|Y}^{(\alpha-1)\theta}(x|Y) \right)^{\frac1{\beta}} \right] \\[0.1cm]
\label{eq: case1-e}
&=\mathbb{E}^{\frac{\beta}{{\beta-1}}} \left[ \left( \sum_{x \in \set{X}}
P_{X|Y}^{\beta}(x|Y) \right)^{\frac1{\beta}} \right] \\[0.1cm]
\label{eq: case1-f}
&= \exp\bigl(-H_{\beta}(X|Y)\bigr).
\end{align}
where the outer expectations are with respect to $Y \sim P_Y$; \eqref{eq: case1-a}
follows from \eqref{eq1: Arimoto - cond. RE} and the equality
$\frac{\beta \rho}{\beta-1} = \frac{\alpha}{\alpha-1}$  (see \eqref{eq: rho_2} and \eqref{eq: rho_1});
\eqref{eq: case1-b} follows from $\tfrac{\beta}{\beta-1}>0$, $\tfrac1{\beta}>0$,
Jensen's inequality and the concavity of $t^{\rho}$ on $[0, \infty)$;
\eqref{eq: case1-c} follows from \eqref{eq: rho_1} and \eqref{eq: rho_2};
\eqref{eq: case1-d} holds due to the concavity of $t^{\theta}$ on $[0, \infty)$
and Jensen's inequality;
\eqref{eq: case1-e} follows from \eqref{eq: rho_2};
\eqref{eq: case1-f} follows from \eqref{eq1: Arimoto - cond. RE}.

{\em Case 2}: If $0 < \beta < \alpha < 1$, then $(\rho, \theta) \in (1, \infty)^2$, and both
$t^{\rho}$ and $t^{\theta}$ are convex. However, \eqref{eq: case1-b} and \eqref{eq: case1-d}
continue to hold since now $\tfrac{\beta}{\beta-1} < 0$.

{\em Case 3}: If $\beta < \alpha < 0$, then $\rho \in (0,1)$ and $\theta \in (1, \infty)$.
Inequality \eqref{eq: case1-b} holds
since $t^\rho$ is concave on $[0, \infty)$ and $\tfrac{\beta}{\beta-1} > 0$;
furthermore, although $t^\theta$ is now convex on $[0, \infty)$, \eqref{eq: case1-d}
also holds since $\tfrac1{\beta}<0$.

\section{Proof of Proposition~\ref{prop: lb/ub}} \label{app: lb/ub}

{\em Proof of Proposition~\ref{prop: lb/ub}\ref{prop: lb/ub - part 1})}:
We need to show that the lower bound on $H_{\alpha}(X|Y)$ in
\eqref{eq: LB on the cond. RE} coincides with its upper bound in
\eqref{eq1: generalized Fano} if and only if
$\varepsilon_{X|Y}=0$ or $\varepsilon_{X|Y}=1-\frac1{M}$. This corresponds, respectively,
to the cases where $X$ is a deterministic function of the observation $Y$
or $X$ is equiprobable on the set $\set{X}$ and independent of $Y$.
In view of \eqref{useful equality}, \eqref{eq1: generalized Fano} and
\eqref{eq: LB on the cond. RE},
and by the use of the parameter $t = \varepsilon_{X|Y} \in \bigl[ 0, 1-\frac1{M} \bigr]$,
this claim can be verified by proving that
\begin{itemize}
\item
if $\alpha \in (0,1)$
and $k \in \{1, \ldots, M-1\}$, then
\begin{align}
\Bigl[ (1-t)^{\alpha} + (M-1)^{1-\alpha} t^{\alpha} \Bigr]^{\frac1{\alpha}}
> \left[k(k+1)^{\frac1\alpha} - k^{\frac1\alpha} (k+1)\right] t + k^{\frac1\alpha+1}-(k-1)(k+1)^{\frac1\alpha}
\label{allerton}
\end{align}
for all $t \in [1-\frac1k, 1-\frac1{k+1})$ with the point $t=0$ excluded if $k=1$;
\item the opposite inequality in \eqref{allerton} holds if $\alpha \in (1,\infty)$.
\end{itemize}
Let the function $v_{\alpha} \colon \bigl[0, 1-\frac1{M} \bigr] \to \Reals$ be defined such that
for all $k \in \{1, \ldots, M-1\}$
\begin{align} \label{v_alpha}
v_{\alpha}(t) = v_{\alpha, k}(t), \quad t \in \bigl[1-\tfrac1{k}, 1-\tfrac1{k+1} \bigr)
\end{align}
where $v_{\alpha, k} \colon \bigl[1-\tfrac1{k}, 1-\tfrac1{k+1} \bigr) \to \Reals$ is given by
\begin{align}
v_{\alpha,k}(t) = \Bigl[ (1-t)^{\alpha} + (M-1)^{1-\alpha} t^{\alpha} \Bigr]^{\frac1{\alpha}}
- \left[k(k+1)^{\frac1\alpha} - k^{\frac1\alpha} (k+1)\right] t - k^{\frac1\alpha+1}+(k-1)(k+1)^{\frac1\alpha}.
\label{v_alpha,k}
\end{align}
Note that $v_{\alpha,k}(t)$ is the difference between the left and right sides of \eqref{allerton}.
Moreover, let $v_{\alpha}(\cdot)$ be continuously extended at $t = 1-\frac1{M}$; it can be
verified that $v_{\alpha}(0) = v_{\alpha}\left(1-\frac1M\right) = 0$.
We need to prove that for all $\alpha \in (0,1)$
\begin{align} \label{v_ineq}
v_{\alpha}(t) > 0, \quad \forall \, t \in \bigl(0, 1-\tfrac1{M} \bigr),
\end{align}
together with the opposite inequality in \eqref{v_ineq} for $\alpha \in (1, \infty)$.
To that end, we prove that
\begin{enumerate}[a)]
\item \label{v properties-1}
for $\alpha \in (0,1) \cup (1, \infty)$, $v_{\alpha}(\cdot)$ is continuous at all points
\begin{align} \label{t_k}
t_k = \tfrac{k}{k+1}, \quad k \in \{0, \ldots, M-1\}.
\end{align}
To show continuity, note that the right continuity at $t_0=0$ and the left continuity at
$t_{M-1}= 1-\tfrac1M$ follow from \eqref{v_alpha}; for $k \in \{1, \ldots, M-2\}$,
the left continuity of $v_{\alpha}(\cdot)$ at $t_k$ is demonstrated by showing that
\begin{align}
\lim_{t \uparrow t_k} v_{\alpha,k}(t) = v_{\alpha,k+1}(t_k),
\end{align}
and its right continuity at $t_k$
is trivial from \eqref{v_alpha}, \eqref{v_alpha,k} and \eqref{t_k};
\item \label{v properties-2}
for all $k \in \{0, 1, \ldots, M-1\}$
\begin{align} \label{v}
v_{\alpha}(t_k) =
\frac{\Bigl[1+(M-1)^{1-\alpha} k^{\alpha} \Bigr]^{\frac1{\alpha}} - (k+1)^{\frac1{\alpha}}}{k+1},
\end{align}
which implies that $v_{\alpha}(\cdot)$ is zero at the endpoints of the interval
$\bigl[0, \frac{M-1}{M} \bigr]$, and if $M \geq 3$
\begin{align}
\label{v1}
& v_{\alpha}(t_k) > 0, \quad \alpha \in (0,1), \\
\label{v2}
& v_{\alpha}(t_k) < 0, \quad \alpha \in (1, \infty)
\end{align}
for all $k \in \{1, \ldots, M-2\}$;
\item \label{v properties-3}
the following convexity results hold:
\begin{itemize}
\item
for $\alpha \in (0,1)$, the function $v_{\alpha}(\cdot)$ is strictly concave
on $[t_k, t_{k+1}]$ for all $k \in \{0, \ldots, M-1\}$;
\item
for $\alpha \in (1, \infty)$, the function $v_{\alpha}(\cdot)$ is
strictly convex on $[t_k, t_{k+1}]$ for all $k \in \{0, \ldots, M-1\}$.
\end{itemize}
These properties hold since, due to the linearity in $t$ of
the right side of \eqref{allerton}, $v_{\alpha}(\cdot)$ is
convex or concave on $\bigl[t_k, t_{k+1}]$ for all $k \in \{0, 1, \ldots, M-1\}$
if and only if the left side of \eqref{allerton} is, respectively, a convex or concave function in $t$;
due to Lemma~\ref{lemma: convexity/ concavity}, the left side of \eqref{allerton} is
strictly concave as a function of $t$ if $\alpha \in (0,1)$ and it is strictly convex if $\alpha \in (1, \infty)$.
\end{enumerate}
Due to Items~\ref{v properties-1})--\ref{v properties-3}), \eqref{v_ineq} holds
for $\alpha \in (0,1)$ and the opposite inequality holds for $\alpha \in (1, \infty)$.
In order to prove \eqref{v_ineq} for $\alpha \in (0,1)$, note that
$ \left[0, 1-\tfrac1M\right] = \bigcup_{i=1}^{M-1} \left[t_{i-1}, t_i\right] $
where $v_{\alpha}(t_0) = v_{\alpha}(t_{M-1})= 0$ and $v_{\alpha}(t_k) > 0$ for every $k \in \{1, \ldots, M-2\}$,
$v_{\alpha}(\cdot)$ is continuous at all points $t_k$ for $k \in \{0, \ldots, M-1\}$,
and $v_{\alpha}$ is strictly concave on $[t_k, t_{k+1}]$ for all $k \in \{0, \ldots, M-2\}$, yielding
the positivity of $v_{\alpha}(\cdot)$ on the interval $(0, 1-\tfrac1M)$. The justification for
the opposite inequality of \eqref{v_ineq}
if $\alpha \in (1, \infty)$ is similar, yielding the negativity of $v_{\alpha}(\cdot)$ on
$(0, 1-\tfrac1M)$.

\par
{\em Proof of Proposition~\ref{prop: lb/ub}\ref{prop: lb/ub - part 2})}:
In view of \eqref{useful equality},
\eqref{eq1: generalized Fano} and \eqref{eq: Barcelona}, for all $\varepsilon_{X|Y} \in [0, \tfrac12)$
and $\alpha \in (0,1) \cup (1, \infty)$:
\begin{align}
& u_{\alpha,M}(\varepsilon_{X|Y})= \frac1{1-\alpha} \, \log \left( (1-\varepsilon_{X|Y})^\alpha
+ (M-1)^{1-\alpha} \varepsilon_{X|Y}^{\alpha} \right), \\[0.2cm]
& l_{\alpha}(\varepsilon_{X|Y})=
\frac{\alpha}{1-\alpha} \, \log \left( 1 + (2^{\frac1\alpha}-2) \varepsilon_{X|Y} \right). \label{l}
\end{align}
Equality \eqref{eq: limit u/l} follows by using L'H\^{o}pital's rule for the calculation of the limit
in the left side of \eqref{eq: limit u/l}.

\section*{Acknowledgment}
The authors gratefully acknowledge the detailed reading of the manuscript by the
Associate Editor and anonymous referees.

\end{document}